\documentclass[a4paper,UKenglish,cleveref, autoref, thm-restate]{lipics-v2019}
\usepackage{xcolor}
\usepackage{thm-restate}
\def\restrict#1{\raise-.5ex\hbox{\ensuremath|}_{#1}}

\title{Efficient Differentially Private $F_0$ Linear Sketching} %TODO Please add

\titlerunning{Efficient Differentially Private $F_0$ Linear Sketching} %TODO optional, please use if title is longer than one line

\author{Rasmus Pagh}{IT University of Copenhagen \and BARC}{pagh@itu.dk}{https://orcid.org/0000-0002-1516-9306}{}%TODO mandatory, please use full name; only 1 author per \author macro; first two parameters are mandatory, other parameters can be empty. Please provide at least the name of the affiliation and the country. The full address is optional

\author{Nina Mesing Stausholm}{IT University of Copenhagen \and BARC}{nimn@itu.dk}{https://orcid.org/0000-0002-4322-7163}{}

\authorrunning{R. Pagh and N. M. Stausholm} %TODO mandatory. First: Use abbreviated first/middle names. Second (only in severe cases): Use first author plus 'et al.'

\Copyright{Rasmus Pagh and Nina Mesing Stausholm} %TODO mandatory, please use full first names. LIPIcs license is "CC-BY";  http://creativecommons.org/licenses/by/3.0/

 %TODO mandatory: Please choose ACM 2012 classifications from https://dl.acm.org/ccs/ccs_flat.cfm 

 %TODO mandatory; please add comma-separated list of keywords

\category{} %optional, e.g. invited paper

\relatedversion{} %optional, e.g. full version hosted on arXiv, HAL, or other respository/website
\supplement{}%optional, e.g. related research data, source code, ... hosted on a repository like zenodo, figshare, GitHub, ...

\funding{This work was supported by Investigator Grant 16582, Basic Algorithms Research Copenhagen (BARC), from the VILLUM Foundation.}%optional, to capture a funding statement, which applies to all authors. Please enter author specific funding statements as fifth argument of the \author macro.

\nolinenumbers %uncomment to disable line numbering

\hideLIPIcs  %uncomment to remove references to LIPIcs series (logo, DOI, ...), e.g. when preparing a pre-final version to be uploaded to arXiv or another public repository

\EventEditors{John Q. Open and Joan R. Access}
\EventNoEds{2}
\EventLongTitle{42nd Conference on Very Important Topics (CVIT 2016)}
\EventShortTitle{CVIT 2016}
\EventAcronym{CVIT}
\EventYear{2016}
\EventDate{December 24--27, 2016}
\EventLocation{Little Whinging, United Kingdom}
\EventLogo{}
\SeriesVolume{42}
\ArticleNo{23}
\begin{document}

\maketitle

\begin{abstract}
A powerful feature of \emph{linear sketches} is that from sketches of two data vectors, one can compute the sketch of the difference between the vectors.
This allows us to answer fine-grained questions about the difference between two data sets. 
In this work we consider how to construct sketches for weighted $F_0$, i.e., the summed weights of the elements in the data set, that are small, differentially private, and computationally efficient. 
Let a weight vector $w\in(0,1]^u$ be given. For $x\in\{0,1\}^u$ we are interested in estimating $\Vert x\circ w\Vert_1$ where $\circ$ is the Hadamard product (entrywise product).

Building on a technique of Kushilevitz et al.~(STOC 1998), we introduce a sketch (depending on $w$) that is linear over GF(2), mapping a vector $x\in \{0,1\}^u$ to $Hx\in\{0,1\}^\tau$ for a matrix $H$ sampled from a suitable distribution $\mathcal{H}$.
Differential privacy is achieved by using \emph{randomized response}, flipping each bit of $Hx$ with probability $p<1/2$.
That is, for a vector $\varphi \in \{0,1\}^\tau$ where $\Pr[(\varphi)_j = 1] = p$ independently for each entry $j$, we consider the \emph{noisy sketch} $Hx + \varphi$, where the addition of noise happens over GF(2).
We show that for every choice of $0<\beta < 1$ and $\varepsilon=O(1)$ there exists $p<1/2$ and a distribution $\mathcal{H}$ of linear sketches of size $\tau = O(\log^2(u)\varepsilon^{-2}\beta^{-2})$ such that:
\begin{enumerate}
\item For random $H\sim\mathcal{H}$ and noise vector $\varphi$, given $Hx + \varphi$ we can compute an estimate of $\Vert x\circ w\Vert_1$ that is accurate within a factor $1\pm\beta$, plus additive error $O(\log(u)\varepsilon^{-2}\beta^{-2})$, w.~p.~$1-u^{-1}$, and
\item For every $H\sim\mathcal{H}$, $Hx + \varphi$ is $\varepsilon$-differentially private over the randomness in $\varphi$.
\end{enumerate}
The special case $w=(1,\dots,1)$ is \emph{unweighted} $F_0$.
Previously, Mir et al.~(PODS 2011) and Kenthapadi et al.~(J.~Priv.~Confidentiality 2013) had described a differentially private way of sketching unweighted $F_0$, but the algorithms for calibrating noise to their sketches are not computationally efficient, either using quasipolynomial time in the sketch size or superlinear time in the universe size $u$.
        
For fixed $\varepsilon$ the size of our sketch is polynomially related to the lower bound of $\Omega\left(\log(u)\beta^{-2}\right)$ bits by Jayram \& Woodruff~(Trans.~Algorithms 2013). The additive error is comparable to the bound of $\Omega\left(1/\varepsilon \right)$ of Hardt \& Talwar~(STOC 2010).
An application of our sketch is that two sketches can be added to form a noisy sketch of the form $H(x_1+x_2) + (\varphi_1+\varphi_2)$, which allows us to estimate $||(x_1+x_2)\circ w||_1$.
Since addition is over GF(2), this is the weight of the symmetric difference of the vectors $x_1$ and $x_2$.
Recent work has shown how to privately and efficiently compute an estimate for the symmetric difference size of two sets using (non-linear) sketches such as FM-sketches and Bloom Filters, but these methods have an error bound no better than $O(\sqrt{\bar{m}})$, %NOTE: RP doesn't think that we known a lower bound.
where $\bar{m}$ is an upper bound on $||x_1||_0$ and $||x_2||_0$.
In particular, our result improves previous work when $\beta  = o\left(1/\sqrt{\bar{m}}\right)$ and $\log(u)/\varepsilon = \bar{m}^{o(1)}$. 

In conclusion our results both improve the efficiency of existing methods for unweighted $F_0$ estimation and extend to a weighted generalization. We also give a distributed streaming implementation for estimating the size of the union between two input streams.
\end{abstract}
 
\section{Introduction}
Estimating the number of distinct values in a set (its \emph{cardinality}), without explicitly enumerating the set, is a classical and important problem in data management.
Sampling-based methods~\cite{haas1995sampling} can in many cases be improved by using algorithms designed with data streams in mind~\cite{kane2010optimal}.
Streaming algorithms based on \emph{linear sketches} can also be used to estimate changes as a data set evolves~\cite{kifer2004detecting} and for approximate query processing in distributed settings~\cite{alon2002tracking, CormodeGHJ12}.
As our first motivating example consider the following SQL query:
\begin{verbatim}
    SELECT P.name
    FROM EMPLOYEES E, HOSPITALIZATION H
    WHERE E.salary > 100000 AND E.name = H.name AND H.year = 2020
\end{verbatim}
\noindent
The size (in bytes) of the query result is a sum weighted by string length over the names that appear in subsets of two relations. That is, estimating the size of the join result is about estimating the \emph{weighted} size of a set intersection.

\medskip

In recent years, \emph{privacy} of database records has become increasingly important when releasing aggregates from a database.
In the example above, the information that a tuple with a particular person exists (and satisfies a certain predicate) can potentially be sensitive.
If the database is distributed, with relations on different servers that are not allowed to expose sensitive information, it is not trivial how to even estimate the join size.

The notion of \emph{differential privacy}~\cite{DworkMNS06} has emerged as the leading approach to providing rigorous privacy guarantees.
It is known that differential privacy comes with pitfalls~\cite{kifer2011no}, but work in the database community has led to privacy-preserving database systems supporting (limited) SQL, see e.g.~\cite{mcsherry2009privacy, wilson2020differentially} and their references.
A challenge in such systems is that the set of queries is often not known ahead of time, so \emph{budgeting} the disclosure of detailed information is highly nontrivial.
An attractive approach to achieving privacy even when faced with unknown queries is to release a summary, or sketch, of the data set from which approximate answers to queries can be computed (as a side effect this also eliminates the need for interaction).
In this paper we consider private linear sketches for the problem of cardinality estimation.

\medskip

{\bf Example.} Suppose that the company Acme Corporation runs an employee satisfaction survey once a year. Management at Acme Corporation made some drastic changes over the past year, and they wish to analyze the impact of these changes on the employees' satisfaction. For a specific improvement, every employee is given a value between 0 and 1, indicating how closely related that improvement is to the employee's work life.
A survey for each improvement is run by a consultant who delivers a summary of the results to the management at Acme Corporation. The consultant ensures that the summary is private, so individual employees cannot be identified from the summary. 
The management at Acme Corporation can combine the summary from last year's survey with the summary from this year's survey to estimate the change in satisfaction over the past year, where the vote of an employee is weighted by the value that employee was given. We note that the summaries should be generated in the same way, but the choice of consultant may change from year to year.
\medskip

More formally, we consider two players that hold sets $A$ and $B$ from a universe $U=\{1,...,u\}$, respectively. For every element $j\in U$ let a fixed, public weight, $w_j\in(0,1]$ be given and for input set $A\subseteq U$ consider the corresponding weight vector $(w_A)_j~=~w_j~\cdot~\mathbf{1}[j\in A]$. The goal is to estimate the weight of the symmetric difference $\Vert w_{A\triangle B}\Vert_1$, in a differentially private manner. We refer the reader to Section~\ref{sec:DPDefinition} for the basics of differential privacy. We may think of the sets as two lists of employees. Given input sets $A$ and $B$, the two players each compute a \emph{linear sketch} of their own set and add noise to obtain privacy as described in Section~\ref{sec:techniques}. These noisy sketches can be thought of as the summaries.

For input sets $A$ and $B$, we note that if we, along with the estimate of the weight of the symmetric difference, have estimates of $\Vert w_A\Vert_1$ and $\Vert w_B\Vert_1$, then we can also estimate $\Vert w_{A\cup B}\Vert_1$, $\Vert w_{A\cap B}\Vert_1$, $\Vert w_{A\backslash B}\Vert_1$ and $\Vert w_{B\backslash A}\Vert_1$ as argued in Section~\ref{sec:merging}. To make this possible, each player also outputs a differentially private version of their set weight. We remark that if all weights $w_j=1$, then the problem reduces to estimating the set \emph{size}, a problem often referred to as $F_0$.

We define and construct a noisy linear sketch over GF(2), the field of size 2, with the following properties:
\begin{itemize}
    \item $\varepsilon$-differentially private
    \item Computationally efficient
    \item Allows estimating the weight of the symmetric difference with small relative error
    \item Space usage is polynomially related to the lower bound (for fixed $\varepsilon$)
\end{itemize}

Previously known results satisfy at most 3 of these properties, see Figure~\ref{table:overview} for an overview. We discuss previous work further in Section \ref{sec:relatedwork}.
Our sketch can be computed and stored for future use, meaning that two players do not have to be active simultaneously but can compute and publish their sketches when they are ready. 
A self-contained description of our linear sketch can be found in Section~\ref{sec:techniques}.
Readers familiar with the sketching literature will realize that our sketch combines a method of Kushilevitz, Ostrovsky, and Rabani~\cite{KushilevitzOR98} with a standard hashing-based subsampling technique (see, e.g.,~\cite{Woodruff14}), and we use a Randomized Response Technique~\cite{RandomizedResponse} with noise parameter $p(\varepsilon)$, to get $\varepsilon$-differential privacy.
Hence, refer to our sketch as the \emph{KOR sketch} and to its noisy counterpart as a \emph{noisy KOR sketch}. 
We note that a related, but non-linear and non-private, sketch has previously been used for estimating size of symmetric difference~\cite{mitzenmacher2014efficient}.
From now on we leave out $\varepsilon$ in the noise parameter and write simply $p$.
We show that the KOR sketch is sufficiently robust to noise to allow precise estimation after adding noise, thus allowing pure differential privacy.

%NOTE: see "Data Streams: Algorithms and Applications" by Muthukrishnan section 5.1.2 and "Comparing Data Streams Using Hamming Norms (How to Zero In)" by Cormode et al. for how we can use JL to estimate cardinality.

%\paragraph*{Techniques.}
We next give an overview of our techniques, discussed in depth in Section~\ref{sec:techniques}. Let $U=\{1,...,u\}$ be the universe from which the input sets are taken. Privacy parameter $\varepsilon$ and accuracy parameter $\beta$ are given, and a sketch size $\tau$ is determined by these parameters. 
We show in Section~\ref{sec:estimateAnalysis} that we can construct an $\varepsilon$-differentially private sketch from which we can compute a $(1+\beta)$-approximation for the weight of the symmetric difference with high probability.

Randomized response~\cite{RandomizedResponse} is applied to the entire sketch $Hx$, meaning that each entry of the sketch is flipped with probability $p < 1/2$. We show in Section~\ref{sec:DPGuarantees} how to choose $p$ as a function of $\varepsilon$ to ensure $\varepsilon$-differential privacy for the sketch. Let $x\circ w$ denote the Hadamard product. Our main theorem is:
\begin{restatable}[Noisy KOR sketch]{theorem}{main}
\label{thm:main}
%For every integer $u > 10$ and 
Let $w\in(0,1]^u$ be given.
For every choice of $0<\beta<1$ and $\varepsilon =O(1)$ there exists a distribution $\mathcal{H}$ over GF(2)-linear sketches
mapping a vector $x\in\{0,1\}^u$ to $\{0,1\}^\tau$, where $\tau~=~O\left(\log^2(u)\varepsilon^{-2}\beta^{-2}\right)$, and a distribution $\mathcal{N}_\varepsilon$ over noise vectors such that: 
\begin{enumerate}
\item For $H\sim\mathcal{H}$ and $\varphi\sim\mathcal{N}_\varepsilon$, given $Hx + \varphi$ we can compute, in time $O(\tau)$, an estimate $\hat{w}$ of $\Vert x\circ w\Vert_1$ that with probability $1-1/u$ satisfies $|\hat{w} - \Vert x\circ w\Vert_1| < \beta \Vert x\circ w\Vert_1 + O\left(\log(u)\varepsilon^{-2}\beta^{-2}\right)$.
\item For every $H$ in the support of $\mathcal{H}$, $Hx + \varphi$ is $\varepsilon$-differentially private over the choice of $\varphi\sim\mathcal{N}_\varepsilon$, and can be computed in time $O(\Vert x\Vert_0 \log(u) + \tau)$, including time for sampling $\varphi$.
\end{enumerate}
%What is the relation between this theorem and the $(1+\beta)$ multiplicative error factor? It seems trivial...
\end{restatable}

\medskip

The assumption that $\varepsilon = O(1)$ is not essential, and is only made to simplify our bounds (which do not improve for privacy parameter $\varepsilon = \omega(1)$).
Without loss of generality we can assume that parameter $\beta$ is such that the error is dominated by $\beta \Vert x\circ w\Vert_1$, because reducing $\beta$ further cannot reduce error by more than a factor~2.
In the unweighted case, setting $\beta = \sqrt[3]{\log(u)/(\varepsilon^2 m)}$ to balance relative and additive error we get error $\tilde{O}(m^{2/3}/\varepsilon^{2/3})$, where the $\tilde{O}$ notations suppresses a polylogarithmic factor. %\textcolor{red}{Should it be $\beta = \sqrt[3]{1 / (m\varepsilon^2)}$? And what is $m$?} Yes!
This is polynomially related to known lower bounds described in section~\ref{sec:lower}.

\begin{figure*}
    \centering
    \makebox[\textwidth]{
    \begin{tabular}{|l|c|c|c|c|c|}
        {\bf Reference} & \begin{tabular}{c}{\bf Diff.}\\{\bf privacy}\end{tabular} &  \begin{tabular}{c}{\bf Additive}\\{\bf error}\end{tabular} & \begin{tabular}{c}{\bf Relative}\\{\bf error}\end{tabular} & \begin{tabular}{c}{\bf Initial.}\\{\bf time}\end{tabular} &  \begin{tabular}{c}{\bf Space}\\{\bf usage}\end{tabular}\\
        \hline
        \hline
        Hardt and Talwar~\cite{HardtT10} & $\varepsilon$ & $\Omega(1/\varepsilon)$ & -- & -- & --\\
        \hline
        McGregor et al.~\cite{mcgregor2010limits} & $\varepsilon$ & $\tilde{\Omega}(\sqrt{m}/e^\varepsilon)$ & -- & -- & --\\
        \hline
        Jayram and Woodruff~\cite{JayramW13} & -- & -- & $1+\beta$ & -- & $\tilde{\Omega}(1/\beta^2)$ \\
        \hline
        \hline
        Kane et al.~\cite{kane2010optimal} & -- & $\tilde{O}(1)$ & $1+\beta$ & $O(1)$ & $\tilde{O}(1/\beta^2)$ \\
        \hline
        Mir et al.~\cite{mir2011pan} & $\varepsilon$ & $\tilde{O}( m^{1-\Omega(1)} / \varepsilon^{O(1)})$ & $1+\beta$ & $\exp((\varepsilon\beta)^{-O(1)})$ & $\tilde{O}((\varepsilon\beta)^{-O(1)})$\\
        \hline
        Kenthapadi et al.~\cite{KenthapadiKMM13} & $(\varepsilon,\delta)$ & $\tilde{O}(\sqrt{m}/\varepsilon)$ & $1+\beta$ & $\tilde{\Omega}(u)$ & $\tilde{O}(1/\beta^{2})^*$ \\
        \hline
        Stanojevic et al.~\cite{StanojevicNY17} & $\varepsilon$ & $\tilde{O}(\sqrt{|A\cup B|}/\varepsilon^2)$ & -- & $\Omega(|A|+|B|)$ &  $\Omega(|A|+|B|)$\\
        \hline
        {\bf This paper} & $\varepsilon$ & $\tilde{O}(m^{2/3}/\varepsilon^{2/3})$ & $1+\beta$ & $\tilde{O}(\varepsilon^{-2}\beta^{-2})$ & $\tilde{O}\left(\varepsilon^{-2}\beta^{-2}\right)$\\
        \hline
    \end{tabular}
    }
    \caption{Selected lower bounds (top part) and upper bounds (bottom part) for estimating the (unweighted) size of the symmetric difference $m = \vert A\triangle B\vert$ from small sketches of sets $A, B \subseteq \{1, \dots, u\}$. Bounds stated as $\tilde{O}$ and $\tilde{\Omega}$ are simplified by suppressing multiplicative factors polynomial in $\log(1/\varepsilon)$, $\log(1/\beta)$, $\log(1/\delta)$, and $\log u$. The non-private bounds in~\cite{JayramW13, kane2010optimal} improve previous results by an $\tilde{O}(1)$ factor, we refer to their references for details. $^*$ The space usage of~\cite{KenthapadiKMM13} is measured in terms of real numbers; it is unclear how much space a private, discrete implementation would need.}
    \label{table:overview}
\end{figure*}

\paragraph*{Applications}
Suppose that Alice holds set $A$ with corresponding characteristic vector $x_A\in\{0,1\}^u$ and Bob holds set $B$ with characteristic vector $x_B\in\{0,1\}^u$. They jointly sample $H\sim\mathcal{H}$ and privately sample $\varphi_A, \varphi_B \sim \mathcal{N}_\varepsilon$ according to Theorem~\ref{thm:main}.
Then $Hx_A + \varphi_A$ and $Hx_B + \varphi_B$ are $\varepsilon$-differentially private.
Furthermore, $(Hx_A + \varphi_B) + (Hx_B + \varphi_B) = (Hx_A + Hx_B) + (\varphi_A + \varphi_B)$, and we show in Section~\ref{sec:merging} that $\varphi_A + \varphi_B ~\sim \mathcal{N}_{\varepsilon'}$ with $\varepsilon' = \varepsilon^2 /(2+2\varepsilon)$.
In Section~\ref{sec:estimateAnalysis} we use this in conjunction with Theorem~\ref{thm:main} to establish:
\begin{restatable}{corollary}{application}
\label{cor:symmetric-difference}
For accuracy parameter $\beta > 0$, consider an $\varepsilon$-differentially private noisy KOR sketch for a set $A$ and an $\varepsilon$-differentially private noisy KOR sketch for a set $B$, based on the same linear sketch $H\sim\mathcal{H}$, sampled independently of $A$ and $B$.
We can compute an approximation $\hat{\Delta}$ of the weight of the symmetric difference, such that with probability $1-1/u$:
\[ \vert \Vert w_{A\triangle B}\Vert_1 - \hat{\Delta}\vert < \beta \Vert w_{A\triangle B}\Vert_1 + \text{\em poly}(1/\varepsilon, 1/\beta, \log u)\enspace .\]
\end{restatable}

In the special case where all weights $w_j$ are 1, this reduces to estimating the \emph{size} of the symmetric difference $A\triangle B$.\\

In Section  \ref{sec:handlingmultisets} we describe how to modify our sketch to apply in a streaming setting. In this case, we estimate the size of the union of the input streams rather than the size of the symmetric difference when merging two sketches.

\section{Related Work}
\label{sec:relatedwork}

In the absence of privacy constraints, seminal estimators for (unweighted) set cardinality that support merging sketches (to produce a sketch of the union) are HyperLogLog \cite{flajolet2007hyperloglog}, FM-sketches~\cite{FlajoletM85}, and bottom-$k$ (aka.~$k$-minimum values) sketches~\cite{bar2002counting}. 
Progress on making these estimators private for set operations include~\cite{TschorschS13} (using FM-sketches) and \cite{SparkaTS18}, which builds a private cardinality estimator to estimate set intersection size using the bottom-$k$ sketch. We note that these sketches do not achieve differential privacy, but are aimed at a weaker notion of privacy. Specifically, they offer a one-sided guarantee that may reveal that an individual element is \emph{not} present in the dataset.
To our best knowledge, a private version of HyperLogLog with provable bounds on accuracy has not been described in the literature.

\medskip

The \emph{weighted} version of cardinality estimation has been less studied.
For (scaled) integer weights in $[W]$ there is a simple reduction that inserts element $i$ with weight $w_i$ by inserting the tuples $(i,1),\dots,(i,w_i)$ into a standard cardinality estimator on the domain $U\times [W]$, but this makes the obtained bounds depend on the number $W$ of possible weights.
Cohen et al.~\cite{cohen2015unified} showed that the class of cardinality estimators that rely on extreme order statistics (for example HyperLogLog) can be efficiently extended to the weighted setting, even for real-numbered weights. 

Note that the weighted $F_0$ estimation problem is different from $F_1$ and $L_1$ estimation in the context of set operations, for example, the union of two identical sets will have the same weighted $F_0$, whereas summing two identical vectors will produce a vector with twice the $L_1$ norm.
In the rest of this section we focus on the standard, unweighted setting.

\subsection{Differentially private cardinality estimators}
Already the seminal paper on pan-privacy~\cite{dwork2010pan} discusses differentially private streaming algorithms for $F_0$ on insertion-only streams.
Their sketch is not linear and does not allow deletions or subtraction of sketches. It is not clear if the sketch can be merged to produce a sketch for the union.
%To get good estimates using these techniques, it is necessary to use space proportional to (an upper bound of) the size of the set whose size is to be estimated.
Recent work by von Voigt et al.~\cite{VoigtT19} has shown how to estimate the cardinality of a set using less space in a differentially private manner using FM-sketches, using the Probabilistic Counting with Stochastic Averaging (PCSA) technique~\cite{FlajoletM85}. These sketches can be merged to obtain a sketch for the union of the input set with a slightly higher level of noise. Privacy is achieved by randomly adding ones to the sketch and by only sketching a sample of the input dataset. 

\emph{Bloom Filters} have been studied extensively to obtain cardinality estimators under set operations (already implicit in~\cite{dwork2010pan}).
Alaggan et al.~\cite{AlagganGMT15} estimated set intersection size by combining a technique for computing similarity between sets, represented by Bloom filters in a differentially private manner, named BLIP (BLoom-then-flIP) filters~\cite{AlagganGK12} with a technique for approximating set intersection of two sets based on their Bloom Filter representation~\cite{BroderM03}. We note that \cite{AlagganGK12} achieves privacy by flipping each bit of the Bloom filter with a certain probability, much like the technique we use to get privacy of our sketch. 
Stanojevic et al.~\cite{StanojevicNY17} show how to estimate set intersection, union and symmetric difference for two sets by computing an estimate for the size of the union, and combined with the size of each set, they show how to compute an estimate for the size of the intersection and the symmetric difference. They achieve privacy by flipping each bit with some probability, like in \cite{AlagganGK12}. 
Also, RAPPOR \cite{ErlingssonPK14} uses Bloom Filters with a Randomized Response technique to collect data from users in a differentially private way but is mainly aimed at computing heavy hitters. 

Though a bound on the expected worst-case error of privately estimating the size of a symmetric difference $|A\triangle B|$ (as in Corollary~\ref{cor:symmetric-difference}) is not stated in any of these papers, an upper bound of $O(\sqrt{\bar{m}})$, where $\bar{m}$ is an upper bound on the size of the sets, follows from the discussion in~\cite{StanojevicNY17} (for fixed $\varepsilon$).
It seems that this magnitude of error is inherent to approaches using Bloom filters since it arises by balancing the error related to the noise and the error related to hash collisions in the Bloom filter.
An advantage and special case of our noisy KOR sketch is that it can be used to directly estimate the size of the symmetric difference, and so the error will depend only on the size of the symmetric difference. It seems that with non-linear sketches it would be necessary to first estimate the size of the union and combine this with the size of each input set  as exhibited in, for example,~\cite{StanojevicNY17}. Hence, the error would depend on the size of the union of the input sets.

%\textcolor{red}{Previously, for example in~\cite{StanojevicNY17, VoigtT19}, non-linear sketches such as FM-sketches and Bloom filters have been used to privately estimate the size of the \emph{union} of sets.
%Such non-linear sketches cannot be directly combined to obtain a sketch for the symmetric difference. 
%Instead, it is necessary to first estimate the size of the union of the input sets and use this together with the size of each input set to estimate the size of the symmetric difference.
%Our sketch achieves smaller error as we compute an estimate for the size of the symmetric difference directly.}

\subsection{Differentially private sketches}
Closely related to our work is the differentially private Johnson-Lindenstrauss (JL) sketch by Kenthapadi et al.~\cite{KenthapadiKMM13}, in which the technique of adding noise to the sketch is also applied. %Kenthapadi et al.~achieve $(\varepsilon,\delta)$-differential privacy, whereas our sketch obtains pure $\varepsilon$-differential privacy. 
Kenthapadi et al.~add Gaussian noise, so to store and maintain a sketched vector, some kind of discretization would be needed (not discussed in their paper).
Discretizing a real-valued private mechanism is non-trivial:
Without sufficient care, one might lose privacy due to rounding in an implementation, as argued by Mironov~\cite{Mironov12}. 
Even if a suitable discretization of the mechanism in~\cite{KenthapadiKMM13} would be possible (see~\cite{canonne2020discrete} for a general discussion), it has several drawbacks compared to our method:
\begin{itemize}
\item It only achieves \emph{approximate} differential privacy as opposed to the pure differential privacy of the noisy KOR sketch.
\item The time needed to update the sketch when a set element is inserted or removed is not constant (in the main method described it is linear in the sketch size).
\item The time needed to initialize the sketch is linear in the size of the sketch matrix, which has $u$ columns, because the noise needs to be calibrated to the sensitivity of the JL sketch matrix, which requires linear time in the size of the sketch matrix. Alternatively, which is the suggestion in Kenthapadi et al., the sketch matrix is assumed to have low sensitivity and noise is calibrated to this sensitivity. If a sketch matrix with a large entry is randomly chosen, the sensitivity of the sketch matrix is large, in which case the noise does not ensure privacy. So with a small probability, privacy is not preserved.
\end{itemize}

Another closely related work is the paper of Mir et al.~\cite{mir2011pan}, which also adds a noise vector after computing standard linear sketches for $F_0$ estimation to make the sketch differentially private. They further initialize their sketches with random noise vectors to also get pan-privacy. %\textcolor{red}{which presents a way of initializing standard linear sketches for $F_0$ estimation to make the sketch differentially private.}
The error bound obtained is similar to ours, and the sketch has a discrete representation, but their method is inferior in terms of time complexity.
This is because they rely on the \emph{exponential mechanism}~\cite{mcsherry2007mechanism}, which is not computationally efficient.
(Note that a preprint of the paper of Mir et al.~\cite{mir2010panArXiv} presented a computationally more efficient method. However, the sensitivity analysis in that paper has an error~\cite{NikolovPersonal2020} that was corrected in the slower method published in~\cite{mir2011pan}.)

Our method is more computationally efficient and arguably simpler than the methods of~\cite{KenthapadiKMM13, mir2011pan}. Our linear sketch is not a replacement for these sketches, though, since our sketch is over GF(2) rather than the reals (or integers).

\subsection{Lower bounds.}\label{sec:lower}
Jayram and Woodruff~\cite{JayramW13} show that, even with no privacy guarantee, to obtain error probability $1/u$ we need a sketch of  $\Omega\left(\log(u)\beta^{-2}\right)$ bits to estimate $F_0$ with relative error $1\pm\beta$.
It is easy to extend this lower bound to our setting, in which an additive error of $c$ is allowed:
Simply insert each item $c$ times, to increase the size of the set so that the additive error is negligible.
Formally this requires us to extend the universe to $U\times \{1,\dots,c\}$, such that the lower bound in terms of the original universe size becomes $\Omega\left(\log(u/c)\beta^{-2}\right)$.
(The reason why we do not use this reduction to eliminate the additive error in our upper bound is that the reduction increases the sensitivity of updates, destroying the differential privacy properties.)

 Hardt and Talwar~\cite{HardtT10} show that an $\varepsilon$-differentially private sketch for $F_0$ must have additive error $\Omega(1/\varepsilon)$, which is comparable (up to polynomial and logarithmic factors) to the additive error we achieve.
%Together with the above, a lower bound for the size of a differentially private sketch with error probability $1/u$ is $\Omega\left(\log(u)\beta^{-2}+1/\varepsilon\right)$ bits. We note that our differentially private sketch is polynomially related to this lower bound.
%

Desfontaines et al.~\cite{DesfontainesLB19} show that it is not possible to preserve privacy in accurate cardinality estimators %\textcolor{red}{that require \emph{idempotence}, i.e., that adding additional copies of an existing element must not change the sketch. It is a necessary requirement for our sketch to work that we are dealing with a set rather than a multiset (i.e., that there are no ``copies'' in the input set).} They also require that one can 
if we can merge several sketches without loss in accuracy. Our sketch will have an increase in noise when merging sketches, and thus does not satisfy the requirement for cardinality estimators formulated in~\cite{DesfontainesLB19}.

McGregor et al.~\cite{mcgregor2010limits} showed that in order to estimate the size of the intersection of two sets $A$ and $B$, based on differentially private sketches of $A$ and $B$, an additive error of $\Omega(\sqrt{u}/e^\varepsilon)$ is needed in the worst case when $A$ and $B$ are arbitrary subsets of $[u]$.
The lower bound holds even in an interactive setting where Alice (holding $A$) and Bob (holding $B$) can communicate, and we require that the communication transcript is differentially private.
The hard input distribution uses sets with symmetric difference of size $\Theta(u)$ with high probability.
Since $|A\cap B| = (|A| + |B| - |A \triangle B|) / 2$, estimating the intersection size is no more difficult (up to constant factors in error) than estimating $|A|$, $|B|$, and $|A \triangle B|$.
We can estimate $|A|$ and $|B|$ with error $O(1/\varepsilon)$ under differential privacy, so it follows that estimating $|A \triangle B|$ under differential privacy requires error $\Omega(\sqrt{u}/e^\varepsilon)$.
For a contrasting upper bound, \cite{Vadhan17, MironovPRV09} suggest an algorithm estimating two-party set intersection size up to an additive error of $O(\sqrt{u}/\varepsilon)$ with high probability.
A lower bound in terms of the size $m$ of the symmetric difference follows by setting $u = m$.

%NOTE: There should be no need to argue that this translates to distinct count estimation. Estimating the inner product is exactly the size of the intersection, but maybe it should be clearer?

\subsection{Noisy sketching.}
In addition to the paper of Mir et al.~\cite{mir2011pan}, there is some previous work on sketching techniques in the presence of noise.
Motivated by applications in learning theory, Awasthi et al.~\cite{awasthi2016learning} considered recovery of a vector based on noisy 1-bit linear measurements.
The resistance to noise demonstrated is analogous to what we show for the KOR sketch, but technically quite different since the linear mapping is computed over the reals before a sign operation is applied.
%NOTE: this paper gives a technique to compress vectors by storing the sign. They wish to show that they can recover the vector from this compression.

In a very recent paper \cite{ChoiDKY20}, Choi et al.~propose a framework for releasing differentially private estimates of various sketching problems in a distributed setting. This framework ensures that the estimates only have a multiplicative error factor. The technique relies on secure multi-party computation and the sketches submitted by each participant are not private and so cannot be released. Further, the results of Choi et al.~do not immediately allow for estimating size or weight of the symmetric difference between two sets.

If the sketching matrix $H$ itself is secret and randomly chosen from a distribution over matrices with entries in a finite field, very strong privacy guarantees on the sketch $Hx$ can be obtained, while still allowing $\Vert x\Vert_0$ to be estimated from $Hx$ with small error~\cite{bioglio2014secure}. Blocki et al.~\cite{BlockiBDS12} prove that the Johnson-Lindenstrauss transform is in fact differentially private, when keeping the sketch matrix secret.
%NOTE: A very recent paper by Choi et al.~\cite{ChoiDKY20}, proposes a framework for releasing differentially private estimates of various sketching problems in a distributed setting. This framework introduces a new way of adding noise: privacy is achieved by multiplying the estimate with well-chosen noise, thus ensuring an estimate with only a multiplicative error factor. Each participant computes the sketch of his or her input, using the same mergeable sketching function, as well as a piece of distributed noise. Applying secure multi-party computation, the parties compute a noisy estimate from the sketches and the noise. The technique by Choi et al.~does not allow for releasing the sketches from each participant to the public as these are not private. Further, the results of Choi et al.~do not immediately allow for estimating size or weight of the symmetric difference between two sets.
However, the condition that the sketch matrix is secret is a serious limitation for applications such as streaming and distributed cardinality estimation that require $H$ to be stored or shared.

%\subsection{\textcolor{red}{Pan-privacy.}}
%
%It was shown in~\cite{dwork2010pan} that it is possible to construct a sketch that is pan-private with respect to a single intrusion.
%\textcolor{red}{However, the space complexity of the sketch described is proportional to $u$, or alternatively to an upper bound $\bar{m}$ on the size of the set.} \textcolor{blue}{This sketch is not linear, and it is not clear how to combine sketches. They cannot be subtracted and so estimating the size of symmetric difference is likely to give a large error.}
%Mir et al.~\cite{mir2011pan} obtained a much smaller space usage, polynomial in $\beta$, $\varepsilon$ and $\log u$.
%However, as mentioned above their method \textcolor{blue}{to obtain differential privacy} is not computationally efficient.

\section{Preliminaries}
We let $[n]=\{1,2,\dots,n\}$ and let $U=[u]$ be the universe that the datasets are taken from. 

For a set $A \subseteq U$, we let $x_A$ denote the characteristic vector for $A$, defined as 
\[
(x_A)_j=\begin{cases}1,\qquad j\in A\\
0,\qquad \text{otherwise} \enspace .
\end{cases}
\]
We write $w_A$ (or $w_{x_A}$) for the weight vector for input set $A$ such that
\[
w_A=x_A\circ w
\]
for fixed, public weights $w_j\in(0,1]$, and $\circ$ denotes the Hadamard product.

For vector $x=(x_1,...,x_u)$ we define $\Vert x\Vert_p=\left(\sum_{j=1}^ux_j^p\right)^{1/p}$ as the $p$-norm of $x$. For $p=0$, we define $\Vert x\Vert_0~=~\sum_{j=1}^u\mathbf{1}[x_j~\neq~0]$, often called the zero-''norm''.
$F_0$ denotes the $0$\textsuperscript{th} frequency moment and represents the number of distinct elements in a stream (or a set). Frequency moments are well-known from the streaming literature, see for example \cite{AlonMS96}. %We let $\Vert x_A\Vert_0$ denote the size of the input set, $A$.

Our sketch $Hx_A$ is comprised of $\log(u)$ ''levels'', $H_ix_A$ for\\ $i~=~0,...,\log(u)-1$. We refer to Section \ref{sec:sketchdescription} for a description of these levels. Let $n$ denote the size of the binary vector representation of $H_ix_A$ for each $i$. Hence, the size of the noisy KOR sketch $Hx_A+\varphi$ is $\tau=n\log u$. Note that $n$ is fixed and depends on the privacy parameter $\varepsilon$ and the accuracy parameter $\beta$. 

Finally, we assume that sets and vectors are stored in a sparse representation, such that we can list the non-zero entries in the input vector $x$ in time $O(\Vert x\Vert_0)$.

\subsection{Hashing-based subsampling}
%In order to allow for merging two sketches by computing the \texttt{XOR}, we need to ensure that if an element exists in both of the players' datasets, the element will cancel itself out when merged using a bit-wise \texttt{XOR}. Hence, either the element should be inserted into both sketches or none of the sketches. 

The sketch matrix $H$ is defined by several hash functions.
For simplicity, we assume access to an oracle representing random hash functions, namely, that we can sample a fully random hash function, and it can be evaluated in constant time.
We do not store the hash function as part of our sketch, so the space for our sketch does not include space required for storing the hash function.
We believe it is possible to replace these hash functions with concrete, efficient hash functions that can be stored in small space while preserving the asymptotic bounds on accuracy, but in order to focus on privacy aspects, we have not pursued this direction.
Importantly, the differential privacy of our method holds for any choice of hash function and does not depend on the random oracle assumption.

To ensure that adding two sketches gives a sketch for the symmetric difference, it is necessary that both players sample the same elements for each $H_i$. To ensure coordinated sampling, we use a hash function, so the same elements from $U$ are sampled by both players. 
We use the following (standard) subsampling technique: 
let $\mathcal{S}$ be the family of all fully random hash functions from $U$ into $[0,1]$. Let $s\sim\mathcal{S}$ uniformly at random. We sample an element $j$ from the input set at level $i=0,...,\log(u)-1$ if and only if $s(j)\in\left(w_j/2^{i+1}, w_j/2^i\right]$.
%let $\mathcal{S}$ be the family of all hash functions from $U$ into $\{0,1\}^{1+\log(u)}$. Let $s\sim\mathcal{S}$ uniformly at random. We sample an element $j$ from the input set at level $i=0,...,\log(u)-1$ if and only if $s(j)$ has exactly $i$ leading zeros.
%
We refer the reader to the survey of Woodruff~\cite{Woodruff14} for more details on subsampling.

\subsection{Differential Privacy}
\label{sec:DPDefinition}
Differential privacy is a statistical property of the behavior of a mechanism \cite{DworkMNS06}. The guarantee is that an adversary who observes the output of a differentially private mechanism will only obtain negligible information about the presence or absence of a particular item in the input data. 
Intuitively, a differentially private mechanism is almost insensitive to the presence or absence of a single element, in the sense that the probability of observing a specific result should be almost the same for any two neighboring sets.

In Definition \ref{def:DP}, we define differential privacy formally in terms of databases. In our application, the databases are sets, and thus \emph{neighboring} means that one set is a subset of the other, and their sizes differ by~1.
\begin{definition}[Differential Privacy \cite{DworkMNS06}]
\label{def:DP}
For $\varepsilon\ge 0$, a randomized mechanism $\mathcal{M}$ is said to be $\varepsilon$-differentially private (or purely differentially private) if for any two neighboring databases, $S$ and $T$ -- i.e., databases differing in a single entry -- and for all $W\subseteq \text{Range}(\mathcal{M})$ it holds that
\[
     \Pr\Big[\mathcal{M}(S)\in W\Big]\le e^{\varepsilon}\cdot\Pr\Big[\mathcal{M}(T)\in W\Big].
\]
For $\varepsilon\ge 0$ and $\delta\in[0,1]$, a randomized mechanism $\mathcal{M}$ is said to be $(\varepsilon,\delta)$-differentially private (or approximately differentially private) if for any two neighboring databases, $S$ and $T$, and for all $W\subseteq \text{Range}(\mathcal{M})$ it holds that
\[
     \Pr\Big[\mathcal{M}(S)\in W\Big]\le e^{\varepsilon}\cdot\Pr\Big[\mathcal{M}(T)\in W\Big]+\delta.
\]
\end{definition}
We show in section \ref{sec:DPGuarantees} that our protocol obtains $\varepsilon$-differential privacy.

Our protocol for estimating the weight of the symmetric difference works in the \emph{local} model of differential privacy, where each player adds noise to their own sketch.
It uses the general technique of achieving privacy by adding noise according to \emph{sensitivity} of a function~\cite{DworkMNS06}.
We note that our sketch would \emph{also} work in a model where vectors supplied by the users are combined using a black-box multi-party \emph{secure aggregation}~\cite{GoryczkaXS13, MelisDC16}.
In this setting, only the sketch for the symmetric difference would be released, and thus, only this sketch would need to be differentially private, meaning that less noise is required.

We can use the Laplace mechanism \cite{DworkMNS06} to get differentially private estimates of the weights of the input sets. These estimates can be used together with an estimate for the weight of the symmetric difference to compute estimates for the union and the intersection of the two input sets with error that is of the same magnitude as the error for estimating the symmetric difference.
For more details about differential privacy, we refer the reader to, for example, \cite{DworkR14}.

%###################### SECURE MULTI-PARTY COMPUTATION ######################
%\subsection{Secure Multi-Party Aggregation}
%Secure multi-party aggregation is a technique that allows a group of mutually distrustful parties, whom each has a sensitive data contribution to compute an aggregate function of these data contributions, without the other parties learning the raw data from the others contributors. 
%In our case, the aggregate function merges sketches from the two contributors into a sketch for the symmetric difference of the input sets from each of the contributors. \textcolor{red}{Omformuler} 
%The technique guarantees that the contributors learn nothing about each other's data contributions apart from what can be inferred from the final result of the aggregation. 

%Noise is added to each contribution before using secure multi-party computation to achieve differential privacy of the final result.

%Secure multi-party computation protocols are likely to be expensive due to the iterated encryption and decryption, as well as the communication of partial results between participating parties. \textcolor{red}{Check -- and which are efficient? We can settle for a two-party. Is that better?}
%We assume honest-but-curious contributors. 

%We will not go into more detail about secure multi-party aggregation, but refer the reader to \textcolor{blue}{ref for textbook fx or Yao's papers?}.
%###################### SECURE MULTI-PARTY COMPUTATION ######################

\section{Techniques}
\label{sec:techniques}

\subsection{Sketch Description}
\label{sec:sketchdescription}
In this section, we describe the noisy KOR sketch in detail. The description is self-contained, but we refer the interested reader to \cite{CormodeGHJ12} for more background on (linear) sketches. 
As mentioned, our sketch combines the techniques from \cite{KushilevitzOR98} with hashing-based subsampling to achieve a sketch that is robust against adding noise, as long as we know how much noise was added. 

We first give the intuition behind the $n\times u$-matrices $H_i$, that our sketch $H$ is comprised of:
Suppose that we have a rough estimate $\hat{E}$ of $\Vert w\Vert_1$, accurate within a constant factor.
Then we can obtain a more precise estimate by sampling (using a hash function) a fraction $n/\hat{E}$ of the elements, for some parameter $n$, and computing the sketch from~\cite{KushilevitzOR98} of size~$n$ for the sampled elements.
This gives an approximation of the number of sampled elements, which in turn gives an approximation of $\Vert w\Vert_1$ with small relative error.
Since we do not know $\Vert w\Vert_1$ within a constant factor -- especially in the setting where we are interested in the size of the symmetric difference -- we use hashing-based subsampling to sample each element $j$ from the input set with probability $w_j/2^{i+1}$ for $i = 0,\dots,\log(u)-1$. Thus for each $i$, we sample elements corresponding to approximately a $1/2^{i+1}$ fraction of the weight and compute the sketch from \cite{KushilevitzOR98} of size $n$ for the sampled elements. For one of these $i$ we are guaranteed to sample approximately a fraction $n/\Vert w\Vert_1$ of the input weight assuming that $\Vert w\Vert_1 > n$. For this $i$, we can obtain a precise estimate of $\Vert w\Vert_1$ from the sketch.

We now define $H_i$ formally.
We first describe the sketch from~\cite{KushilevitzOR98} as a linear sketch over GF(2).
Let $\mathcal{F}$ be the family of all hash functions from universe $U$ into $[n]$, and pick $h\sim \mathcal{F}$ uniformly at random. 
The hash function $h$ uniquely defines an $n\times u$-matrix $K$, where
\[
K_{k,j}=\begin{cases}
1,\qquad \text{if $h(j)=k$}\\
0,\qquad \text{otherwise} \enspace .
\end{cases}
\] We combine this with the following sampling technique:

Let $\mathcal{S}$ be the family of all hash functions from $U$ to $[0,1]$. Sample $s\sim \mathcal{S}$ uniformly at random. The hash function $s$ defines a $u\times u$-diagonal matrix $S_i$ for each $i=0,...,\log(u)-1$, defined by
\[
(S_i)_{j,j}=\begin{cases}
1,\qquad \text{if $s(j)\in\left(w_j/2^{i+1},w_j/2^i\right]$}\\
0,\qquad \text{otherwise} \enspace .
\end{cases}
\]
The matrix-vector product $S_ix$ represents subsample of input vector $x$, where we sample each element with probability $w_j/2^{i+1}$.

%##################### NOTE TO SELF #######################
% Usually we get a distribution by sampling uniformly at random from a family.
%##################### NOTE TO SELF #######################

We are finally ready to define $H_i$ as $H_i=K S_i$, which is an $n\times u$-matrix over GF(2).
By definition:
\[
(H_i)_{k,j}=\begin{cases}
1,\qquad (h(j)=k)\land (s(j)\in\left(w_j/2^{i+1},w_j/2^i\right])\\
0,\qquad \text{otherwise} \enspace .
\end{cases}
\]
The KOR sketch can be represented as an $n\log(u)\times u$-matrix $H$, formed by stacking $H_1,...,H_{\log(u)}$.

Let $\mathcal{N}_\varepsilon$ be a distribution over vectors from $\{0,1\}^{n\log(u)}$, where each entry is 1 independently with probability $p$. We show in Section \ref{sec:estimateAnalysis} that it suffices to set $p =1/(2+\varepsilon)$.
%, where the function $p(\varepsilon)$ chosen in Section~\ref{sec:DPGuarantees}. 
Sample the noise (or \emph{pertubation}) vector $\varphi\sim \mathcal{N}_\varepsilon$ independently and uniformly at random.
The \emph{noisy} KOR sketch of $x$ is then computed (over GF(2)) as:
\[
Hx+\varphi.
\]

\subsection{Estimation}
\label{sec:protocol}
Next, we describe how to compute a weight estimate from a sketch $Hx+\varphi$. Let $w$ be the weight vector associated with $x$. Let $\varphi_i$ be the restriction of $\varphi$ to the entries that are added to $H_ix$ when adding $\varphi$ to $Hx$.
%as described in Section~\ref{sec:sketchdescription} for
To compute an estimate for $\Vert w\Vert_1$, for each $i=0,...,\log(u)-1$ count the number of 1s in $H_i x+\varphi_i$, $Z_i=\Vert H_ix+\varphi_i\Vert_0$ and compute the interval:
\begin{align}
\label{tech:interval}
I_i=\begin{cases}
[0,u]\qquad\qquad\qquad\qquad\qquad\qquad\qquad\qquad\quad\ \ \text{if $Z_i\ge(1-\gamma)n/2$}\\
\left[2^in\ln\left(\frac{\frac{1}{2/\varepsilon+1}}{1-\frac{2Z_i}{(1+\gamma)n}}\right), 2^in\ln\left(\frac{\frac{1}{2/\varepsilon+1}}{1-\frac{2Z_i}{(1-\gamma)n}}\right)\right]\qquad\ \text{otherwise.}
\end{cases}
\end{align}
where
$\gamma<\frac{\beta-1/n}{7e^3(2/\varepsilon+1)}.$
%$\gamma=\beta\left(\frac{1}{\frac{2}{\varepsilon}+1}\right)/\left(2(3+\beta)e^2\right)$.
Compute the intersection $I=\bigcap_{i=0}^{\log(u)-1}I_i$ and check if the maximum value in $I$ is within a factor $(1+\eta)$ of the minimum value in $I$ for 
\[
\eta = \frac{6\gamma\left(e^3\left(\frac{2}{\varepsilon}+1\right)-1\right)}{1+\gamma-2\gamma\left(e^3\left(\frac{2}{\varepsilon}+1\right)\right)}%=\frac{6\left(1-\frac{1}{e^3\left(\frac{2}{\varepsilon}+1\right)}\right)}{\frac{7}{\beta-1/n}-\left(2-\frac{1}{e^3\left(\frac{2}{\varepsilon}+1\right)}\right)}
\enspace .
\]
If that is the case, every element in $I$ is a good estimate for $\Vert w\Vert_1$ (having relative error at most $(1+\beta)$) with high probability. 
Otherwise, $\Vert w\Vert_1$ is small with high probability, and we let the estimate for $\Vert w\Vert_1$ be 0.
We analyze the accuracy of this estimator in Section~\ref{sec:proof}.

\subsection{Application to symmetric difference}
\label{sec:merging}
In this section, we describe a differentially private protocol to compute an estimate for the weight of the symmetric difference between sets held by two parties.
First, we show that the sum of two noisy KOR sketches, $Hx_A+\varphi$ and $Hx_B+\psi$, is a noisy KOR sketch for the symmetric difference, $H(x_{A\triangle B})+(\varphi+\psi)$, which has the same properties as $Hx_A+\varphi$ and $Hx_B+\psi$, but for $\varepsilon'<\varepsilon$ as more noise is added.

\begin{lemma}
\label{lem:merging}
Adding two noisy KOR sketches with perturbation vectors $\varphi\sim \mathcal{N}_\varepsilon$ and $\psi\sim \mathcal{N}_\varepsilon$, respectively, will yield a noisy KOR sketch for the symmetric difference of the input sets with noise $\varphi+ \psi\sim \mathcal{N}_{\varepsilon'}$ for  $\varepsilon'=\varepsilon^2/(2+2\varepsilon)$.
\end{lemma}
\begin{proof}
Let $x_A$ and $x_B$ be the input vectors from each of the two players. Let $H$ be as defined in Section~\ref{sec:sketchdescription}, and define $\varphi, \psi$ as the noise vectors for the noisy KOR sketches for $x_A$ and $x_B$, respectively. 
We have (over GF(2)) that
\begin{align*}
\left(Hx_A+\varphi\right)+ \left(Hx_B+\psi\right)&=\left(Hx_A+Hx_B\right)+ \left(\varphi+\psi\right)\\&=H(x_A+x_B)+ \left(\varphi+\psi\right).
\end{align*}
This is exactly the noisy KOR sketch for the symmetric difference with perturbation $\varphi+ \psi$. Note that we observe a 1 in an entry of $\varphi+ \psi$ with probability $p'=p(1-p)+(1-p)p=2p(1-p)$.
We show in Section~\ref{sec:estimateAnalysis} that we can let $p=\frac{1}{2+\varepsilon}$. Observe that
\[
p'=\frac{1}{2+\varepsilon'}=\frac{2}{2+\varepsilon}\left(1-\frac{1}{2+\varepsilon}\right)
\]
which implies that $\varepsilon'=\varepsilon^2/(2+2\varepsilon)$.
\end{proof}

By Lemma~\ref{lem:merging} we can treat a sketch for the symmetric difference exactly like a sketch for input vector $x$ although with a different privacy parameter $\varepsilon'$. Hence, Theorem~\ref{thm:main} gives us Corollary \ref{cor:symmetric-difference}, restated here for convenience:

\application*

Note that the additive error in Corollary \ref{cor:symmetric-difference} still depends polynomially on $\varepsilon$ even for privacy parameter $\varepsilon'$, which is explained by the fact that $\varepsilon'=\varepsilon^2/(2+2\varepsilon)$.

%What happens to the error when we get $\varepsilon'$ instead of $\varepsilon$? (polynomial in $1/\varepsilon$) State that we pick $p=1/2-\varepsilon/2\log(u)$ and refer to sec. ?
%Write our what happens for $p'$ and $\varepsilon'$.

%Note that we need also argue about the running time. It is not exactly clear how fast this is when using the matrix representation of the hash function. 

Finally, we assumed that $\Vert w_A\Vert_1$ and $\Vert w_B\Vert_1$ were released with Laplacian noise, which gives an expected additive error of $O(1/\varepsilon)$ for each of $\Vert w_A\Vert_1$ and $\Vert w_B\Vert_1$ \cite{DworkMNS06}. We can use the following equations to get estimates for the union, intersection and difference:
\begin{align*}
&\Vert w_{A\cup B}\Vert_1 =\frac{\Vert w_A\Vert_1 +\Vert w_B\Vert_1 + \Vert w_{A\triangle B}\Vert_1}{2},\\
&\Vert w_{A\cap B}\Vert_1 = \frac{\Vert w_A\Vert_1 +\Vert w_B\Vert_1 - \Vert w_{A\triangle B}\Vert_1}{2}\\ &\Vert w_{A\backslash B}\Vert_1=\frac{\Vert w_A\Vert_1+\Vert w_{A\triangle B}\Vert_1-\Vert w_B\Vert_1}{2} \enspace.
\end{align*}
\noindent
That is, the error is bounded by half the error of the estimate of the symmetric difference size plus $O(1/\varepsilon)$.
%Finally, we assumed that $\vert A\vert$ and $\vert B\vert$ were released with Laplacian noise, which gives an expected additive error of $O(1/\varepsilon)$ for each of $\vert A\vert$ and $\vert B\vert$ \cite{DworkMNS06}. We can use the following equations to get estimates for the union, intersection and difference:
%\begin{align*}
%&\vert A\cup B\vert =\frac{\vert A\vert +\vert B\vert + \vert A\triangle B\vert}{2}, \qquad
%\vert A\cap B\vert = \frac{\vert A\vert +\vert B\vert - \vert A\triangle B\vert}{2}\\ &\vert A\backslash B\vert=\frac{\vert A\vert+\vert A\triangle B\vert-\vert B\vert}{2} \enspace.
%\end{align*}
%
%\noindent
%That is, the error is bounded by half the error of the estimate of the symmetric difference size plus $O(1/\varepsilon)$.

\section{Proof of Theorem~\ref{thm:main}}\label{sec:proof}

In this section we give a proof of Theorem~\ref{thm:main}, restated here for convenience:

\main*

\subsection{Noise level and Differential Privacy Guarantees}
\label{sec:DPGuarantees}
We first show that the noisy KOR sketch $H x+\varphi$ satisfies $\varepsilon$-differential privacy, which proves part 2 of Theorem~\ref{thm:main}. Intuitively, removal/insertion of a single element can change only a single entry in the sketch, as the element is inserted into only a single level.%\textcolor{red}{Additionally, the weights are in $(0,1]$ --  that is, our estimator has sensitivity $1$. By the postprocessing property of differential privacy, it suffices to show that $Hx_A+\varphi$ is differentially private.}

% NOTE TO SELF: If the sketch is DP, then any estimate computed from the sketch is also DP. And so, the estimate for the weight will also be DP.

\begin{restatable}{lemma}{DPguarantee}
\label{lem:dpguarantee}
    If $p\in \left(\frac{1}{e^{\varepsilon}+1},\tfrac{1}{2}\right)$ then $H x+\varphi$ is $\varepsilon$-differentially private.
\end{restatable}
\begin{proof}[Proof Sketch]
The proof follows from the privacy of the Randomized Response Technique~\cite{RandomizedResponse}. To make this work self-contained, we included a full proof in Appendix \ref{app:privacy}.
\end{proof}

\subsection{Bounding accuracy}
\label{sec:estimateAnalysis}

In this section, along with Section~\ref{sec:sumup}, we prove the first part of Theorem~\ref{thm:main}.
Let an input vector $x$ be given and define $w$ to be the corresponding weight vector. We will mainly consider each $H_ix$ isolated, so let $\varphi_i$ be the $n$-dimensional (binary) randomness vector as described in the proof of Lemma \ref{lem:dpguarantee}.
First, we state two useful lemmas.

\begin{restatable}{lemma}{expectations}
\label{lem:expectations}
For each $i=0,...,\log(u)-1$ let  $L_i=\Vert H_i x\Vert_0$ and $Z_i=\Vert H_i x + \varphi_i\Vert_0$. Then:
\begin{align}
\label{exp:numoneslinearsketch}
    \mathop{\operatorname{E}}_{\substack{h\sim\mathcal{F},\\s\sim\mathcal{S}}}[L_i]=\frac{n}{2}\left(1-\prod_{j\in A}\left(1-\frac{w_j}{2^in}\right)\right)
\end{align}
\begin{align}
\label{exp:numonesnoisysketch}
    \mathop{\operatorname{E}}_{\substack{h\sim\mathcal{F},\\s\sim\mathcal{S},\\\varphi_i\sim\mathcal{N}_p}}[Z_i]=\frac{n}{2}\left(1-\left(1-2p\right)\prod_{j\in A}\left(1-\frac{w_j}{2^in}\right)\right)
\end{align}
\end{restatable}
\begin{proof}
We refer the reader to Appendix \ref{app:expectations} for the proof.
\end{proof}

\begin{restatable}{lemma}{concentrated}
\label{lem:concentrated}
For $i=0,...,\log(u)-1$ let $Z_i=\Vert H_i x + \varphi_i\Vert_0$. For any $0<\gamma<1$, we have with probability at least $1-6\log(u)e^{-\frac{\gamma^2p^3n}{6^2\cdot 3}}$ that for all $i=0,...,\log(u)-1$ simultaneously: \[(1-\gamma)\mathop{\operatorname{E}}_{\substack{h\sim\mathcal{F},\\s\sim\mathcal{S},\\\varphi_i\sim\mathcal{N}_p}}[Z_i]<Z_i<(1+\gamma)\mathop{\operatorname{E}}_{\substack{h\sim\mathcal{F},\\s\sim\mathcal{S},\\\varphi_i\sim\mathcal{N}_p}}[Z_i].\]
\end{restatable}
\begin{proof}
We refer the reader to Appendix \ref{app:concentration} for the proof.
\end{proof}

%From equation (\ref{exp:numonesnoisysketch}) we obtain an expression for $m$ in terms of $\operatorname{E}[Z_i]$:
%\[
%m=\ln\left(\frac{1-2p}{1-2\mathop{\operatorname{E}}[Z_i]/n}\right)/\ln\left(\frac{1}{1-\frac{1}{2^{i}n}}\right).
%\]
%Note that $m$ is monotone in $\operatorname{E}[Z_i]$, so we can insert the bounds from Lemma~\ref{lem:concentrated}, to obtain the following interval for $m$ which holds for all $i$ with probability at least $1-6\log(u)e^{-\gamma^2p^3n/108}$:
%\begin{align}
%\label{intervalaroundm}
%m\in\left[\ln\left(\frac{1-2p}{1-2\frac{Z_i}{(1+\gamma)n}}\right)/\ln\left(\frac{1}{1-\frac{1}{2^{i}n}}\right), \ln\left(\frac{1-2p}{1-2\frac{Z_i}{(1-\gamma)n}}\right)/\ln\left(\frac{1}{1-\frac{1}{2^{i}n}}\right)\right]=I_i(p).
%\end{align}
%(((This definition of $I_i$ is different from equation (1), use $I_i(p)$, or similar?)))
First, we consider the case when $1<n<\Vert w\Vert_1$. In Lemma~\ref{lem:accuracy} we state that in this case, with high probability we get an error of at most a factor $(1+\beta)$ for a well-chosen $\gamma$, where $\gamma$ is a function of the privacy parameter $\varepsilon$, the accuracy parameter $\beta$ and the size of the universe, $u$.
For convenience, define 
\begin{align}
\label{tech:intervalforp}
I_i(p)=\begin{cases}
[0,u]\qquad\qquad\qquad\qquad\qquad\qquad\qquad\qquad\quad\ \text{if $Z_i\ge (1-\gamma)n/2$}\\
\left[2^in\ln\left(\frac{1-2p}{1-\frac{2Z_i}{(1+\gamma)n}}\right), 2^in\ln\left(\frac{1-2p}{1-\frac{2Z_i}{(1-\gamma)n}}\right)\right]\qquad \text{otherwise}
\end{cases}
\end{align}
and $\hat{w}:=2^in\ln\left(1/\prod_{j\in A}\left(1-\frac{w_j}{2^in}\right)\right)$.
We prove our result in two steps:
\begin{enumerate}
    \item If $\hat{w}\in I_i(p)$ for all $i=0,...,\log(u)-1$, then there is some $i$ such that any value from (\ref{tech:intervalforp}) estimates $\hat{w}$ up to a factor $(1+\eta)$, where $\eta$ is a function of $\gamma$ and $\varepsilon$. 
    \item $\Vert w\Vert_1\le \hat{w}\le \left(1+\frac{1}{2^in}\right)\Vert w\Vert_1$ for each $i$. Specifically, $\Vert w\Vert_1\le \hat{w}\le \left(1+\frac{1}{n}\right)\Vert w\Vert_1$ for all $i$.
\end{enumerate}
Hence, we choose $\gamma$ independent of $i$ such that $(1+\eta)\left(1+\frac{1}{n}\right)~\le~ (1~+~\beta)$ for at least one of the intervals $I_i(p)$. We pick $\gamma$ to work for the $i$ where $ \Vert w\Vert_1/(2^in)\in[1, 2)$ as this corresponds to having an input of size between $n$ and $2n$ (we obtain this input size by the sampling from $x$ in $H_i$). If $\Vert w\Vert_1\ge n$, there is such an $i$, and we can identify it by checking that the endpoints of the interval are sufficiently close together, as described in Section \ref{sec:protocol}. 
We consider the case when $\Vert w\Vert_1<n$ in Section~\ref{sec:sumup} where we show that in this case, the error is bounded by an additive factor of $O(n)$.

%\textcolor{red}{Specifically, we choose $\gamma$ such that $m\in I_i(p)$ for all $i$, then there is an $i$, such that any value from $I_i(p)$ is an estimate of $m$ with at most a multiplicative error of $(1+\beta)$. The observation is that we can bound the size of the interval by restricting ourselves to a single choice of $i$ and choose $\gamma$ depending on that $i$. We use the $i$ where $1\le \Vert w\Vert_1/(2^in)\le 2$ as this corresponds to having an input of size between $n$ and $2n$ (we obtain this input size by the sampling from $x$ in $H_i$). We consider the case when $\Vert w\Vert_1<n$ in Section~\ref{sec:sumup} where we show that in this case, the error is bounded by an additive factor of $O(n)$.}
\begin{restatable}{lemma}{accuracy}
\label{lem:accuracy}
Assume $\Vert w\Vert_1 > n > 1$, and $\beta>\frac{1}{n}$. With probability at least $1-6\log(u)e^{-\frac{\gamma^2p^3n}{108}}$ there exists an $i\in\{0,...,\log(u)-1\}$ such that any element from $I_i(p)$ is a $(1+\beta)$-approximation to $\Vert w\Vert_1$ for
\[
\gamma<\frac{\left(\beta-\frac{1}{n}\right)(1-2p)}{7e^3}\enspace .%}.
\] Specifically, $i$ where $\frac{\Vert w\Vert_1}{2^in}\in[1,2)$, gives these guarantees.
\end{restatable}
\begin{proof}[Proof Sketch]
We give an informal sketch of the proof and refer the reader to Appendix \ref{app:accuracy} for the formal proof.
We first remark that for $\gamma$ as described, Lemma \ref{lem:concentrated} implies that if $\Vert w\Vert_1/(2^in)\le 2$, then $Z_i<(1-\gamma)n/2$ with high probability. Hence, it suffices to consider the intervals from (\ref{tech:intervalforp}) of the form $I_i(p)=\left[2^in\ln\left(\frac{1-2p}{1-\frac{2Z_i}{(1+\gamma)n}}\right), 2^in\ln\left(\frac{1-2p}{1-\frac{2Z_i}{(1-\gamma)n}}\right)\right]$.
Define 
\[
\hat{w}:=2^in\ln\left(\frac{1}{\prod_{j\in A}\left(1-\frac{w_j}{2^in}\right)}\right).
\]
From Lemma \ref{lem:expectations}, we have 
\[
\prod_{j\in A}\left(1-\frac{w_j}{2^in}\right)=\frac{1-\frac{2\operatorname{E}[Z_i]}{n}}{1-2p}.
\]
Assume that the bounds in Lemma~\ref{lem:concentrated} are satisfied. We remove this assumption shortly. By the bounds in Lemma~\ref{lem:concentrated}, $\hat{w}\in I_i(p)$ for all $i$. 
We show that $\hat{w}$ is contained in an interval, which is slightly bigger than $I_i(p)$ whenever $\Vert w\Vert_1/(2^in)\in[1,2)$ and show that the endpoints of this interval are within a factor $(1+\eta)$ of each other, where $\eta$ is a function of $\gamma$. Clearly, then $I_i(p)$ is also sufficiently small for this $i$. Denote this interval $I_i^*(p)$. Any element from $I_i^*(p)$ is a $(1+\eta)$-approximation to $\hat{w}$. 
%The idea is to use the bounds from Lemma~\ref{lem:concentrated} again to compute a slightly bigger interval around $\hat{w}$, $I_i^+(p) \supset I_i(p)$ whose end-points depend on $\mathop{\operatorname{E}}[Z_i]$ rather than $Z_i$. 
Removing the assumption that the bounds in Lemma~\ref{lem:concentrated} hold, we simply get a small error probability and conclude that with probability at least $1-6\log(u)e^{-\gamma^2p^3n/108}$ we have $\hat{w}\in I_i(p)$ for all $i$, and thus any value from $I_i^*(p)$ is a $(1+\eta)$ estimation to $\hat{w}$ with high probability.
Observing that $\Vert w\Vert_1~<~\hat{w}~\le~ \left(1+\frac{1}{n}\right)\Vert w\Vert_1$ for any $i$, we choose $\gamma$ in terms of $\beta$ such that $(1+\eta)\left(1+\frac{1}{n}\right)<(1+\beta)$. Then any value from $I_i^*(p)$ is a $(1+\beta)$-approximation for $\Vert w\Vert_1$. We formally choose $\gamma$ in Appendix \ref{app:accuracy}.
We remark that the assumption $\Vert w\Vert_1/(2^in)\in[1,2)$ allows us to choose $\gamma$ independent of $i$, such that we can compute $I_i(p)$ for all $i$ with a single value of $\gamma$.
%\textcolor{red}{we now have an interval which we know is small enough for a specific $i$, but we haven't selected $\gamma$ yet.}The value $\gamma$ is chosen to ensure that the endpoints of $I_i^+(p)$ are within a factor $(1+\eta)$ of each other at the level $i$ where $\Vert w\Vert_1/(2^in)\in[1,2)$. We formally choose $\gamma$ in Appendix \ref{app:accuracy}. Clearly, then $I_i(p)$ is also sufficiently small for this $i$, and so any element from $I_i(p)$ will be a $(1+\eta)$-approximation to $\hat{w}$.
\end{proof}

%\textcolor{blue}{We wish to let $p=\frac{1}{2}(1-\varepsilon)$. Note that \begin{align*}
%    \frac{1}{2}(1-\varepsilon)>\frac{1}{e^\varepsilon+1}\qquad &\Leftrightarrow \qquad (e^\varepsilon+1)(1-\varepsilon)=e^{\varepsilon}(1-\varepsilon)+(1-\varepsilon)>2\\
%    &\Leftrightarrow \qquad e^{\varepsilon}(1-\varepsilon)>1+\varepsilon
%\end{align*}} \textcolor{red}{This is wrong. This only works for $\varepsilon>1$... It didn't work before either. Can we maybe use something like $1/(2+\varepsilon)$ because $e^\varepsilon>1+\varepsilon$ for $\varepsilon>0$?} 
Observing that $\frac{1}{2+\varepsilon}>\frac{1}{e^\varepsilon+1}$ for $\varepsilon>0$, we let $p=1/\left(2+\varepsilon\right)$ and observe that for $I_i:=I_i\left(1/\left(2+\varepsilon\right)\right)$ with the choice of $\gamma$ described in Lemma \ref{lem:accuracy}, we get the interval $I_i$ in (\ref{tech:interval}). %\textcolor{red}{should it be bound together better?}
%(((Need a unique definition.)))

\subsection{Putting things together}
\label{sec:sumup}
In this section we consider the accuracy in the remaining case where $\Vert w\Vert_1\leq n$.
We also analyze the running time.
Combining with Section~\ref{sec:DPGuarantees} this completes the proof of Theorem~\ref{thm:main}.

%Let $p=\frac{1}{2}-\frac{\varepsilon}{2\log(u)}$ and thus
%\[
%\gamma=\frac{\beta\varepsilon}{2(3+\beta)e^2\log(u)}.
%\]
%Then we get the $I_i$ from (\ref{intervalaroundm}) in terms of $\varepsilon$ and $\beta$ (for readability we leave $\gamma$ in the expression), which shows the choice of interval in  (\ref{tech:interval}):
%\begin{align*}
%\left[\ln\left(\frac{\varepsilon}{\log(u)\left(1-\frac{2Z_i}{\left(1+\gamma\right)n}\right)}\right)/\ln\left(\frac{1}{1-\frac{1}{2^{i}n}}\right), \ln\left(\frac{\varepsilon}{\log(u)\left(1-\frac{2Z_i}{\left(1-\gamma\right)n}\right)}\right)/\ln\left(\frac{1}{1-\frac{1}{2^{i}n}}\right)\right]=I_i.
%\end{align*}
Note that if $\varepsilon>1$, we can start our protocol by dividing $\varepsilon$ by a suitable constant, $c$ such that $\varepsilon'=\varepsilon/c<1$. Changing $\varepsilon$ by a constant will change our bounds by a constant factor as well.  Hence, we can without loss of generality assume $\varepsilon<1$. We can also, without loss of generality, assume $u>10$ -- this will at most increase the failure probability and space by a constant factor.

We first show a sufficient upper bound on the sketch size $\tau = n\log u$.
Observe that $p>1/4$ and let $c_\gamma= 7e^3$ be a constant. %If it works for u=11 then it works for u=5.
Then we want
$e^{-\frac{\gamma^2p^3n}{108}}<1/u^2$
as this ensures a failure probability of at most
$6\log(u)/u^2<1/u.$ Noting that 
\[
(1-2p)^2=\left(1-\frac{2}{2+\varepsilon}\right)^2=\left(\frac{1}{2/\varepsilon+1}\right)^2=\frac{1}{4/\varepsilon^2+4/\varepsilon+1}>\frac{\varepsilon^2}{20},
\] we have
\begin{align*}
e^{-\frac{\gamma^2p^3n}{108}}&<e^{-\frac{\left(\frac{\left(\beta-\frac{1}{n}\right)(1-2p)}{7e^3}\right)^2n/4^3}{108}}=e^{-\frac{\left(\beta-\frac{1}{n}\right)^2 \left(\frac{1}{(2/\varepsilon+1)^2}\right)n}{4^3\cdot c_{\gamma}^2\cdot 108}}\\&<e^{-\frac{\left(\beta-\frac{1}{n}\right)^2\varepsilon^2 n}{20\cdot 4^3c_\gamma^2\cdot 108}}<1/u^2
\end{align*}
when letting $n=O\left(\log(u)\beta^{-2}\varepsilon^{-2}\right)$.
%\[
%e^{-\frac{\gamma^2p^3n}{108}}<e^{-\frac{\left(\frac{\beta \varepsilon/\log(u)}{2(3+\beta)e^2}\right)^21/4^3n}{108}}=e^{-\frac{\beta^2\varepsilon^2n}{(\log(u)c_\gamma)^2\cdot4^3\cdot  108}}<1/u^2,
%\]
%as this ensures a failure probability of at most
%$6\log(u)/u^2<1/u.$
%(((Start proof by assuming wlog that $u\ge 9$?)))
%It is enough to choose $n>\frac{\log^3(u)c'_\gamma}{\beta^2\varepsilon^2}$
%for $c_\gamma'=2\cdot 4^3\cdot c_\gamma^2\cdot 108$.
%For constant $c'_\gamma=c_\gamma^2\cdot 27\cdot 64$ this implies \textcolor{red}{it is enough to choose $n$ like this.} (((why? is this not just a sufficient condition on $n$?)))
%\[
%n>\frac{\log^3(u)c'_\gamma}{\varepsilon^2\beta^2}.
%\]
Hence, the size of the sketch is 
\[
\tau=\log(u)\cdot n=O\left(\frac{\log^2(u)}{\varepsilon^2\beta^2}\right).
\]
Note that this $n$ satisfies the requirement $\beta>1/n$ from Lemma \ref{lem:accuracy}.

We argue about the error: 
Note that if $\Vert w\Vert_1\ge n$, then if one of the intervals $I_i$ is sufficiently small and $\hat{w}\in I_i$ for all $i=0,...,\log(u)-1$, then $\hat{w}\in I= \bigcap_{i=0}^{\log(u)-1}I_i$ and $I$ is also sufficiently small to give the wanted estimate. So by Lemma \ref{lem:accuracy}, we can check if the endpoints of $I$ are within a factor at most $(1+\eta)$ of each other, and if so, with probability $1-1/u$ any value from $I$ is within a factor $(1+\beta)$ of $\Vert w\Vert_1$. If $I$ is too big, then none of the intervals $I_i$ was sufficiently small implying that our assumption that $\Vert w\Vert_1/(2^in)\in[1,2)$ does not hold for any $i$. Hence, with probability $1-1/u$ we have $\Vert w\Vert_1<n$. We refer to the formal proof in Appendix \ref{app:accuracy} for the details.
Our protocol sets the estimate of $\Vert w\Vert_1$ to 0 leading to an additive error of $O(n)$ when $I$ was too big. This means that we get an additive error of at most $n=O\left(\log(u)\beta^{-2}\varepsilon^{-2}\right)$, as required.\\
%NOTE: Let $I_Z$ be the interval that we measure and let $I_{\beta^*}$ be the bigger interval. There are three cases: 1) both $I_Z$ and $I_{\beta^*}$ are small enough, 2) $I_Z$ is small enough, $I_{\beta^*}$ is not, 3) none of them is small enough. Case 1) is clear. In case 2) are in the situation where $\Vert w\Vert_1/2^in>2$ for the $i$ in question. But we still pick an estimate which is within the approximation factor. It doesn't matter how we get there, as long as we get an estimate. In case 3) we observe that $I_Z$ is not small enough, which implies that $I_{\beta^*}$ is not small enough. This can happen in two ways: a) the concentration bounds do not hold, b) $\Vert w\Vert_1$ is too small. If I is small enough, then we know with high probability that any value is a good estimate for $\hat{w}$ (this fails if the concentration bounds don't hold). 

Finally, we comment on the running times:
For the first part of Theorem~\ref{thm:main}, we note that in order to compute the estimate, we need to count the number of ones in $H_ix+\varphi_i$ for each $i=0,...,\log(u)-1$, compute the intervals $I_i$ and their intersection and check if it is sufficiently small. Counting the number of ones in all $H_ix+\varphi_i$ is the bottleneck and requires time $O(\tau)$. 
For the second part of Theorem~\ref{thm:main}, note that we can initialize the randomness vector $\varphi$ in time $O(\tau)$ and we can hash vector $x$ in time $O\left(\Vert x\Vert_0\log(u)\right)$ assuming that we can iterate over $x$ in time $O(\Vert x\Vert_0)$.

Combining with Lemma~\ref{lem:accuracy} and Lemma~\ref{lem:dpguarantee}, we have completed the proof of Theorem~\ref{thm:main}. %we conclude that with probability at least $1-\frac{1}{u}$ we get an estimate $\hat{m}$ of $\Vert x\Vert_0$ that satisfies $|\hat{m} - \Vert x\Vert_0| < \beta \Vert x\Vert_0 + O(\log^3(u)\varepsilon^{-2}\beta^{-2})$ for a sketch of size $\tau=O\left(\log^4(u)\beta^{-2}\varepsilon^{-2}\right)$.

\section{Distributed Streaming Implementation}
\label{sec:handlingmultisets}
In a streaming setting want a sketch which can be updated and two sketches can be merged to give a sketch for the union of the input streams, while we cannot guarantee that there are no duplicates in the input stream.
In this case, our sketch does not immediately apply, as items with an even number of occurrences would ''cancel out''. Such items would therefore never be represented in the sketch, as the sketch is over GF(2). This issue can easily be fixed: the idea is to add another layer of sampling, such that we sample each \emph{occurrence} of a data item with probability $1/2$. Hence, we treat identical items independently on each occurrence and so ensures that an entry in the sketch is 1 with probability $1/2$, regardless of the number of copies of identical items and collisions with other items. We refer to this as the \emph{pre-sampled} sketch. The intuition is that the number of copies of an item inserted in the pre-sampled sketch is even or odd with probability $1/2$. By Chernoff bounds the fraction of elements that are sampled an odd number of times is very close to $1/2$ with high probability. Thus it is natural to consider the estimator that is two times the estimator described in Section \ref{sec:protocol}.

To understand this in more detail we argue that merging two  (non-private) pre-sampled sketches over $GF(2)$ gives a sketch for the \emph{union} of the two input sets. Suppose $z\in A\cup B$, $h(z)=k$ and that $z$ is sampled at level $i$. We argue that $\Pr[(H_ix_{A\cup B})_k=1]=1/2$.
Note that 
\[(H_ix_{A\cup B})_k=1\qquad \Leftrightarrow\qquad (H_ix_{A})_k\neq (H_ix_{B})_k.
\]
Further, we have that if $z\in A$, then $\Pr[(H_ix_A)_k=1]=1/2$ regardless of the number of other elements hashing to $k$ at level $i$. If no elements from $A$ hash to entry $k$ at level $i$, then $\Pr[(H_ix_A)_k~=~1]~=~0$. We have
\begin{align*}
\Pr[(H_ix_{A\cup B})_k=1]&=\Pr[(H_ix_{A})_k=1]\Pr[(H_ix_{B})_k=0]\\&\qquad+\Pr[(H_ix_{A})_k=0]\Pr[(H_ix_{B})_k=1],
\end{align*}
which is $1/2$ whenever $z\in A\cup B$.
%NOTE TO SELF: Be aware that we could actually get away with less noise in this case, something like p=1/(2\vareps). This is because we just need to "simulate" an almost uniform distribution with our noise (because whenever something was inserted, the value is uniform (prob 1/2 each) from {0,1}, while if nothing was inserted, the value is 0. Hence, in order to hide something we just need noise enough to make it look uniform.

%%
%% Bibliography
%%

%% Please use bibtex, 

\medskip

{\bf Acknowledgement.} We thank Shuang Song and Abhradeep Guha Thakurta for feedback on a previous version of this manuscript.

\newpage

\bibliography{bibl}

\appendix

\section{Omitted proofs}
\label{app:leftoutproofs}

\subsection{Differential Privacy Guarantees}
\label{app:privacy}
\DPguarantee*
\begin{proof}
    Let $A$ and $B$ be two neighboring input sets with corresponding characteristic vectors, $x_A$ and $x_B$, where neighboring means that one set is a subset of the other and the sizes differ by 1. By symmetry of differential privacy, we can without loss of generality assume  that $A$ is the smaller set. Suppose that $B\backslash\{z\}=A$. The element $z$ can only affect $H_ix$ for $i$ where $z$ is sampled. If $z$ is never sampled, then $H x_A = H x_B$ and privacy is trivial. So assume $i\in \{0,...,\log(u)-1\}$ such that $s(z)\in\left(w_z/2
   ^{i+1},w_z/2^i\right]$. We limit our attention to $H_ix_A+\varphi_i$, where we can think of $\varphi_i$ as the restriction of the $n\log(u)$-dimensional random vector $\varphi\sim\mathcal{N}_\varepsilon$ to the entries that would be added to $H_i x_A$ when adding $\varphi$ to $H x_A$. We show that $H_i x_A+\varphi_i$ is $\varepsilon$-differentially private. This implies that the entire sketch, $H x_A+\varphi$, is $\varepsilon$-differentially private. %We note that if $z$ is not sampled, that is, if $s(z)$ does not have exactly $i$ leading zeros, then $z$ will not contribute to $H_ix_B$ and so the sketches are identical which trivially implies differential privacy. The only level where $z$ can  
    %Note that $z$ might still contribute to $H_jx_B$ for $j<i$. Hence, we assume that $z$ is sampled at level $i$. %(((This is true but requires an argument, namely that the definition of DP is symmetric with respect to A and B.)))
    
    Inserting $z$ into the sketch implies that $H_ix_A$ and $H_ix_B$ will differ in exactly one entry, i.e., $\Vert H_ix_A+H_ix_B\Vert_0=1$. %\textcolor{red}{Let $\varphi_i,\psi_i$ be $n$-dimensional (binary) randomness vectors, where we restrict the $n\log(u)$-dimensional randomness vectors $\varphi,\psi\sim\mathcal{N}_\varepsilon$ to the entries corresponding to the $i$\textsuperscript{th} level. We can think of $\varphi$ as the concatenation of $\log(u)$ $n$-dimensional randomness vectors $\varphi_1,...,\varphi_{\log(u)}$.} 
    Fix a noisy sketch, $S_i$. There exist unique vectors $\varphi_i$ and $\psi_i$, such that $S_{i}=H_ix_A+\varphi_i=H_ix_B+\psi_i$. Note that $\Vert \varphi_i -\psi_i\Vert_0=1$.
    Let $\Vert \varphi_i\Vert_0=r$. Then $\Vert \psi_i\Vert_0~=~r'$ for $r'\in\{r+1, r-1\}$. Conditioned on $\Vert \varphi_i\Vert_0~=~r$ and $\Vert \psi_i\Vert_0~=~r'$, the probabilities of randomly drawing exactly these randomness vectors are, respectively: 
    \[
    (1-p)^{n-r}p^r\qquad \text{and}\qquad (1-p)^{n-r'}p^{r'}.
    \]
    
    Let $\varepsilon=O(1)$ be given. By Section~\ref{sec:DPDefinition} it is enough to show that for any fixed output
     $S_{i}~=~H_ix_A+\varphi_i~=~H_ix_B+\psi_i$, we have 
    \[
    e^{-\varepsilon}\le \frac{\Pr\big[\text{observe $S_i$ from $A$}\big]}{\Pr\big[\text{observe $S_i$ from $B$}\big]}=\frac{\Pr\big[\text{observe $H_ix_A+\varphi_i$ from $A$}\big]}{\Pr\big[\text{observe $H_ix_B+\psi_i$ from $B$}\big]}\le e^{\varepsilon}.
    \]
    where the probability is over the randomness in $\varphi_i$ and $\psi_i$. The sketches for $A$ and $B$ are computed using the same $H_i$, so the choice of $H_i$ has no impact.
    
	%If the additional element is sampled, it will flip a bit in the level-$i$ sketch. (((More clear and succinct to state this in terms of the vector $H_ix_A + H_ix_B$ having a single 1.))) Something like $\Vert H_ix_A + H_ix_B\Vert_0=1$? So the number of 1s in $H_ix_B$ will differ from the number of 1s in $H_ix_A$ by at most one. Note also that the rest of the 1s will be in the same positions in both level-$i$ sketches. %Let $S_{i}$ be the fixed level-$i$ output sketch observed. We argue about the probabilities of observing $S_{i}$ from $H_ix_A$ and $H_ix_B$.
	%\textcolor{blue}{Make sure that this is clear and mathify:}
    %Let $\varphi_i,\psi_i$ be $n$-dimensional randomness vectors, where we restrict $\varphi,\psi\sim\mathcal{N}_\varepsilon$ to the entries corresponding to the $i$\textsuperscript{th} level. There is exactly one $n$-dimensional randomness vector for each of $H_ix_A$ and $H_ix_B$ that will result in $S_{i}$ and that these randomness vectors will differ in exactly one entry. %(((Define the randomness vectors before making this statement, to make it statable in mathematical notation?)))
    %Let $\varphi_i$ be the randomness vector such that $H_ix_A+\varphi_i=S_{i}$ and similarly $\psi_i$ be the randomness vector such that $H_ix_B+\psi_i=S_{i}$. Let $\Vert \varphi_i\Vert_0=r$. Then $\Vert \psi_i\Vert_0=r'$ where $r'\in\{r+1, r-1\}$. The probability of randomly drawing exactly these randomness vectors is: 
    %\[
    %(1-p)^{n-r}p^r\qquad \text{and}\qquad (1-p)^{n-r'}p^{r'},
    %\]
    %respectively.
    Hence, to obtain differential privacy it suffices that for every possible value of $r$ and $r'\in\{r+1, r-1\}$
    \[
    e^{-\varepsilon}\le \frac{(1-p)^{n-r}p^r}{(1-p)^{n-r'}p^{r'}}=\frac{1}{(1-p)^{r-r'}p^{r'-r}}\le e^{\varepsilon},
    \]
    which is satisfied for 
    $1/2> p\ge 1/\left(e^{\varepsilon}+1\right)$,
    since $p<1/2$ by assumption.
\end{proof}

\subsection{Expectations}
\label{app:expectations}
\expectations*
\begin{proof}
Let $A$ be the input set with corresponding weight vector $w$. 
Let $v_i\in\mathbb{Z}_{\ge 0}^n$ be a vector such that for each $k\in[n]$
\[
\left(v_i\right)_k=\sum_{j\in A}\mathbf{1}\left[\frac{s(j)}{w_j}\in\left(1/2^{i+1},3/2^{i+1}\right]\right]\cdot \mathbf{1}\left[h(j)=k\right].
\]
That is, each entry $(v_i)_k$ is the number of candidates for entry $k$ in the sketch at level $i$, i.e., the number of items $j$ that hash to $k$ and satisfy $\frac{s(j)}{w_j}\in\left(1/2^{i+1},3/2^{i+1}\right]$. Since $s(j)$ is uniform, we have for such a candidate
%NOTE TO SELF: This corresponds to having selected part of $s(j)$, but not specified the exact precision yet. When considering the next digits for $s(j)$, we can tell what part of this interval that $s(j)$ is in. 
\[
\Pr_{s\sim\mathcal{S}}\left[s(j)\in\left(w_j/2^{i+1},2w_j/2^{i+1}\right]\ \Big\vert\ s(j)\in\left(w_j/2^{i+1},3w_j/2^{i+1}\right]\right]=\frac{1}{2}.
\]
If there is at least one candidate for entry $k$ then, by the Principle of Deferred Decisions, the probability that we sample an odd number of these is $1/2$ and so for $i=0,...,\log(u)-1$
%#######################################################################
%INTUITION -- PRINCIPLE OF DEFERRED DECISIONS: 
%The elements that have at least $i$ leading zeros are candidates to be sampled. We sample the ones where the $i+1$st bit in $s(j)$ is a 1. For any sequence of choices we've seen so far (among the other candidates), the probability that the last candidate is sampled is $1/2$. At any point in time, regardless of whether the result based on the sequence, we've seen so far (the result is the contents of the $k$th entry in the sketch and is either a 0 or a 1), the probability that the next element will make it odd is $1/2$ (if it were even, it becomes odd if we sample the next element, which happens with probability 1/2. If it were odd, it stays odd if we don't sample the next, which also happens with probability $1/2$).
%#######################################################################
%we see that for $i=0,...,\log(u)-1$
\begin{align*}
&\Pr_{\substack{h\sim\mathcal{F}\\s\sim\mathcal{S}}}[(H_ix_A)_k=1\ \vert\ (v_i)_k\neq 0]=\frac{1}{2},\\ &\Pr_{\substack{h\sim\mathcal{F}\\s\sim\mathcal{S}}}[(H_ix_A)_k=1\ \vert\ (v_i)_k= 0]=0.
\end{align*}
As 
\[
\Pr_{s\sim S}\left[\frac{s(j)}{w_j}\in\left(1/2^{i+1},3/2^{i+1}\right]\right]=\Pr_{s\sim S}\left[s(j)\in\left(w_j/2^{i+1},3w_j/2^{i+1}\right]\right]%=\int_{\frac{w_j}{2^{i+1}}}^{\frac{3w_j}{2^{i+1}}}1 dx
=\frac{w_j}{2^i},
\] 
we have 
\[
\Pr_{\substack{h\sim\mathcal{F}\\s\sim\mathcal{S}}}[(v_i)_k\neq 0]=1-\prod_{j\in A}\left(1-\frac{w_j}{2^in}\right).
\]
We conclude that 
\[
\Pr_{\substack{h\sim\mathcal{F},\\s\sim\mathcal{S}}}[(H_ix_A)_k=1]=\frac{1-\prod_{j\in A}\left(1-\frac{w_j}{2^in}\right)}{2}.
\]
and letting $L_i=\sum_{k=1}^n(H_ix_A)_k$, we get
\[
\mathop{\operatorname{E}}_{\substack{h\sim\mathcal{F},\\s\sim\mathcal{S}}}[L_i]=\frac{n}{2}\left(1-\prod_{j\in A}\left(1-\frac{w_j}{2^in}\right)\right)
\]
We similarly compute an expression for $\mathop{\operatorname{E}}_{\substack{h\sim\mathcal{F},s\sim\mathcal{S},\varphi_i\sim\mathcal{N}_p}}[Z_i]$. Let $\varphi_i$ be the restriction of a randomness vector $\varphi\sim\mathcal{N}_\varepsilon$ to the entries that are added to $H_ix_A$ when adding $\varphi$ to $Hx_A$.
We see that 
\begin{align*}
    &\Pr_{\substack{h\sim\mathcal{F},\\s\sim\mathcal{S},\\\varphi_i\sim\mathcal{N}_p}}\Big[(H_ix_A+\varphi_i)_k=1\Big]\\&=\Pr_{\substack{h\sim\mathcal{F},\\s\sim\mathcal{S},\\\varphi_i\sim\mathcal{N}_p}}\Big[(H_ix_A+\varphi_i)_k=1\ \vert\ (H_ix_A)_k=1\Big]\cdot \Pr_{\substack{h\sim\mathcal{F},\\s\sim\mathcal{S}}}\Big[(H_ix_A)_k=1\Big]\\&+\Pr_{\substack{h\sim\mathcal{F},\\s\sim\mathcal{S},\\\varphi_i\sim\mathcal{N}_p}}\Big[(H_ix_A+\varphi_i)_k=1\ \vert\ (H_ix_A)_k=0\Big]\cdot \Pr_{\substack{h\sim\mathcal{F},\\s\sim\mathcal{S}}}\Big[(H_ix_A)_k=0\Big]\\
    &=(1-p)\cdot \Pr_{\substack{h\sim\mathcal{F},\\s\sim\mathcal{S}}}\Big[(H_ix_A)_k=1\Big]+p\cdot \Pr_{\substack{h\sim\mathcal{F},\\s\sim\mathcal{S}}}\Big[(H_ix_A)_k=0\Big]\\&=(1-p)\cdot \frac{1}{2}\left(1-\prod_{j\in A}\left(1-\frac{w_j}{2^in}\right)\right)+p\cdot \left(1-\frac{1-\prod_{j\in A}\left(1-\frac{w_j}{2^in}\right)}{2}\right)\\
    &=\frac{1}{2}-\left(\frac{1}{2}-p\right)\prod_{j\in A}\left(1-\frac{w_j}{2^in}\right)
\end{align*}
showing that 
\begin{align*}
    \mathop{\operatorname{E}}_{\substack{h\sim\mathcal{F},\\s\sim\mathcal{S},\\\varphi_i\sim\mathcal{N}_p}}[Z_i]=\frac{n}{2}\left(1-\left(1-2p\right)\prod_{j\in A}\left(1-\frac{w_j}{2^in}\right)\right).
\end{align*}
\end{proof}

\subsection{Concentration bounds}
\label{app:concentration}
\concentrated*
Before proving Lemma~\ref{lem:concentrated}, we mention the following lemma:
\begin{restatable}{lemma}{hoeffding}
\label{lem:hoeffding}
Let $L_i=\Vert H_ix_A\Vert_0$.
For any $0<\gamma'<1$, we have with probability at least $1-4\log(u)e^{-2\gamma'^2n}$
\[
\mathop{\operatorname{E}}_{\substack{h\sim\mathcal{F},\\s\sim\mathcal{S}}}[L_i]-2\gamma' n\le L_i\le \mathop{\operatorname{E}}_{\substack{h\sim\mathcal{F},\\s\sim\mathcal{S}}}[L_i]+2\gamma' n
\]
for all $i=0,...,\log(u)-1$ simultaneously.
\end{restatable}
\begin{proof}
Let $A$ be the input set and $w$ the corresponding weight vector.
Let $v_i\in\mathbb{Z}_{\ge 0}^n$ be a vector such that for each $k\in[n]$
\[
\left(v_i\right)_k=\sum_{j\in A}\mathbf{1}\left[\frac{s(j)}{w_j}\in\left(1/2^{i+1},3/2^{i+1}\right]\right]\cdot \mathbf{1}\left[h(j)=k\right]
\] so $(v_i)_k$ is the number of candidates for entry $k$ in the sketch at level $i$.
Let $V_i=\Vert v_i\Vert_0=\sum_{k=1}^n\mathbf{1}[(v_i)_k\neq 0]$. $V_i$ is a sum of negatively associated random variables (for definition and argument see Section 4.1 in \cite{dubhashi2009concentration}), so by Theorem 4.3 in \cite{dubhashi2009concentration}, we can use the Hoeffding bound to see that with probability at least $1-2e^{-2n\gamma'^2}$ we have for any $i=0,...,\log(u)-1$
\begin{align}
\label{eq:numentrieshit}
\operatorname{E}[V_i]-\gamma' n\le V_i\le \operatorname{E}[V_i]+\gamma' n. 
\end{align}
Let $L_i=\Vert H_ix_A\Vert_0=\sum_{k=1}^n(H_ix_A)_k$ denote the number of ones in the linear sketch.
For fixed $V_i$, $L_i$ is a sum of independent random variables with (by the principle of deferred decisions)
\[\Pr\Big[(H_ix_A)_k=1\ \vert\ (v_i)_k\neq 0\Big]=\frac{1}{2},\qquad
\Pr\Big[(H_ix_A)_k=1\ \vert\ (v_i)_k= 0\Big]=0.\]
%\textcolor{red}{maybe include short argument}
So for any fixed $V_{i}=t$
\begin{align}
\label{eq:expectedfixedv}
\operatorname{E}\Big[L_i\ \vert\ V_{i}=t\Big]=\frac{t}{2}.
\end{align}
Furthermore, as $L_i$ is a sum of independent random variables for a fixed choice of $V_i$, we can use the Hoeffding bound: with probability at least $1-2e^{-2n\gamma'^2}$
\[
\operatorname{E}\Big[L_i\ \vert\ V_i=t\Big]-\gamma' n\le L_i\restrict{V_i=t}\le \operatorname{E}\Big[L_i\ \vert\ V_i=t\Big]+\gamma' n,
\]
where $L_i\restrict{V_i=t}$ means the value of $L_i$ when we assume that $V_i=t$.
Combining this with (\ref{eq:numentrieshit}) and (\ref{eq:expectedfixedv}) a union bound gives with probability at least $1-4e^{-2n\gamma'^2}$
\begin{align}
\label{ineq:boundingLiVi}
\frac{\operatorname{E}[V_i]-\gamma'n}{2}-\gamma'n\le L_i\le \frac{\operatorname{E}[V_i]+\gamma'n}{2}+\gamma'n.
\end{align}
Simultaneously, (\ref{eq:numentrieshit}) and (\ref{eq:expectedfixedv}) gives
\begin{align}
\label{ineq:boundingexpLi}
\frac{\operatorname{E}[V_{i}]-\gamma' n}{2} \le \operatorname{E}[L_i]\le \frac{\operatorname{E}[V_{i}]+\gamma' n}{2},
\end{align}
which implies
\begin{align}
\label{ineq:boundingexpLibyVi}
2\operatorname{E}[L_i]-\gamma'n\le \operatorname{E}[V_i]\le 2\operatorname{E}[L_i]+\gamma'n.
\end{align}
Note that in the union bound from (\ref{ineq:boundingLiVi}), we already assumed that (\ref{eq:numentrieshit}) was satisfied, so (\ref{ineq:boundingexpLibyVi}) is trivially satisfied under the union bound without changing the probability guarantees. 
Hence, inserting (\ref{ineq:boundingexpLibyVi}) into (\ref{ineq:boundingLiVi}), we have 
\begin{align}
\frac{2\operatorname{E}[L_i]-2\gamma'n}{2}-\gamma'n\le L_i\le \frac{2\operatorname{E}[L_i]+2\gamma'n}{2}+\gamma'n.
\end{align}
which finally shows that with probability at least $1-4e^{-2n\gamma'^2}$ we have 
\[
\operatorname{E}[L_i]-2\gamma'n\le L_i\le \operatorname{E}[L_i]+2\gamma'n.
\]
A union bound over the $\log(u)$ values of $i$ concludes the proof.
\end{proof}

We are now ready to prove Lemma~\ref{lem:concentrated}.
\begin{proof}[Proof of Lemma~\ref{lem:concentrated}]
Fix $i$. Let $L_i=\Vert H_ix_A\Vert_0$ and $Z_i=\Vert H_ix_A+\varphi_i\Vert_0$. We let $Z_i\restrict{L_i=t}$ be the number of ones in $H_ix_A+\varphi_i$, assuming that $L_i=t$.
For any fixed value $t\in\{0,...,n\}$ of $L_i$, we have 
\begin{align}
\label{equal:onesboundedont}
\mathop{\mathop{\operatorname{E}}}_{\substack{\varphi_i\sim\mathcal{N}_p}}\Big[Z_i\restrict{L_i=t}\Big]=(1-p)\cdot t+p(n-t)=np+(1-2p)t.
\end{align}
By Lemma~\ref{lem:hoeffding}, with probability at least $1-4\log(u)e^{-2\gamma'^2n}$ we have for all $i=0,...,\log(u)-1$
%#################### NOTE TO SELF #########################
% Because $t$ will be bounded by $\operatorname{E}[H_i]\pm\gamma'n$ with this probability.
%$Z_i\restrict{H_i}$ denotes the number of ones when we hold the linear sketch fixed and only consider the perturbation.
%#################### NOTE TO SELF #########################
\begin{align}
\label{boundingexpZilower}
&\mathop{\operatorname{E}}_{\substack{h\sim\mathcal{F},\\s\sim\mathcal{S},\\\varphi_i\sim\mathcal{N}_p}}[Z_i]\ge np+(1-2p)\left(\mathop{\mathop{\operatorname{E}}}_{\substack{h\sim\mathcal{F},\\s\sim\mathcal{S}}}[L_i]-2\gamma'n\right)\\
\label{boundingexpZiupper}
&\mathop{\operatorname{E}}_{\substack{h\sim\mathcal{F},\\s\sim\mathcal{S},\\\varphi_i\sim\mathcal{N}_p}}[Z_i]\le np+(1-2p)\left(\mathop{\operatorname{E}}_{\substack{h\sim\mathcal{F},\\s\sim\mathcal{S}}}[L_i]+2\gamma'n\right)
\end{align}

Furthermore, for any fixed $H_i$, let $Z_i\restrict{H_i}$ denote the number of ones in $H_ix_A+\varphi_i$, conditioned on this choice of $H_i$. We note that fixing $H_i$ is equivalent to fixing $L_i$ as $L_i$ is uniquely determined by $H_i$ and the input. $Z_i\restrict{H_i}$ is a sum of independent random variables, where the randomness comes from the perturbation. So for any $0<\gamma^*<1$, a Chernoff bound gives
\begin{align}
&\Pr_{\substack{\varphi_i\sim\mathcal{N}_p}}\Bigg[Z_i\restrict{H_i}>(1+\gamma^*)\mathop{\operatorname{E}}\left[Z_i\restrict{H_i}\right]\ \lor\ Z_i\restrict{H_i}<(1-\gamma^*)\mathop{\operatorname{E}}\left[Z_i\restrict{H_i}\right]\Bigg]\\
\label{chernoff}
&\le 2e^{-\gamma^{*2}\mathop{\operatorname{E}}\left[Z_i\restrict{H_i}\right]/3}
\end{align}
where $\mathop{\operatorname{E}}\left[Z_i\restrict{H_i}\right]$ is over $\varphi_i\sim\mathcal{N}_p$. By (\ref{equal:onesboundedont}), $\mathop{\operatorname{E}}_{\substack{\varphi_i\sim\mathcal{N}_p}}\left[Z_i\restrict{H_i}\right]\ge np$ for any choice of $H_i$, so $2e^{-\gamma^{*2}pn/3}$ is an upper bound on (\ref{chernoff}). Moreover, (\ref{chernoff}) holds for all $i=0,...,\log(u)-1$ simultaneously with probability at most $2\log(u)e^{-\gamma^{*2}pn/3}$. We conclude that 
\begin{align}
%\label{chernoffpositive}
    &\Pr_{\substack{\varphi_i\sim\mathcal{N}_p}}\Bigg[\forall i: (1-\gamma^*)\mathop{\operatorname{E}}\left[Z_i\restrict{H_i}\right]<Z_i\restrict{H_i}<(1+\gamma^*)\mathop{\operatorname{E}}\left[Z_i\restrict{H_i}\right]\Bigg]\\
    \label{chernoffpositive}
    &\ge 1-2\log(u)e^{-\gamma^{*2}pn/3}
\end{align}

%#################### NOTE TO SELF #########################
%So: 
%For fixed $H_i$, we have with probability $1-2\log(u)e^{-\gamma'^2pn/3}$ for all $i$
%\[
%(1-\gamma)\left(np+(1-2p)t\right)=(1+\gamma)\operatorname{E}\left[Z_i\restrict{L_i=t}\right]\le Z_i\restrict{L_i=t}\le (1+\gamma)\operatorname{E}\left[Z_i\restrict{L_i=t}\right]=(1+\gamma)\left(np+(1-2p)t\right).
%\] So with probability at least $1-4\log(u)e^{-\gamma'^2pn/3}$ we can also insert the bounds from (\ref{boundingexpZi})
%#################### NOTE TO SELF #########################

Combining (\ref{boundingexpZilower}), (\ref{boundingexpZiupper}) and (\ref{chernoffpositive}) and letting $\gamma'=\gamma^*$, we have by a union bound that for all levels $i$ simultaneously, where the expectation is over $h\sim\mathcal{F}$ and $s\sim\mathcal{S}$
\begin{align*}
Z_i&\ge (1-\gamma')\left(np+(1-2p)\left(\mathop{\operatorname{E}}[L_i]-2\gamma' n\right)\right)\\
Z_i&\le (1+\gamma')\left(np+(1-2p)\left(\mathop{\operatorname{E}}[L_i]+2\gamma' n\right)\right),
\end{align*}
with probability at least \[1-\left(4\log(u)e^{-2n\gamma'^2}+2\log(u)e^{-\gamma'^2pn/3}\right)\ge 1-6\log(u)e^{-\gamma'^2pn/3}.\]
By Lemma~\ref{lem:expectations}, this is equivalent to
\begin{align}
\label{ineq:pickinggamma'lower}
Z_i&\ge (1-\gamma')\left(\mathop{\operatorname{E}}[Z_i]-2(1-2p)\gamma' n\right)\\
\label{ineq:pickinggamma'upper}
Z_i&\le (1+\gamma')\left(\mathop{\operatorname{E}}[Z_i]+2(1-2p)\gamma' n\right).
\end{align}
where the expectation is over $h\sim\mathcal{F}, s\sim\mathcal{S}$ and $\varphi_i\sim\mathcal{N}_p$.
We pick a suitable $\gamma'$:
\begin{align*}
\gamma' = \frac{\gamma p}{6}\qquad &\Rightarrow\qquad 2(1-2p)\gamma' n= (1-2p)\frac{\gamma p}{3}n \\&\Rightarrow\qquad 2(1-2p)\gamma' n\le \frac{\gamma(1-2p)}{3}\mathop{\operatorname{E}}[Z_i].
\end{align*}
%\textcolor{red}{Why do we pick this in particular? It seems to also work for }
Hence, let $\gamma'= \frac{\gamma p}{6}$.
Inserting into (\ref{ineq:pickinggamma'lower}) and (\ref{ineq:pickinggamma'upper}) we have
\begin{align*}
Z_i&\ge \left(1-\frac{\gamma p}{6}\right)\left(\mathop{\operatorname{E}}[Z_i]-\frac{\gamma(1-2p)}{3}\mathop{\operatorname{E}}[Z_i]\right)\\ Z_i&\le \left(1+\frac{\gamma p}{6}\right)\left(\mathop{\operatorname{E}}[Z_i]+\frac{\gamma(1-2p)}{3}\mathop{\operatorname{E}}[Z_i]\right)
\end{align*}
where $\mathop{\operatorname{E}}[Z_i]$ is over $h\sim\mathcal{F}, s\sim\mathcal{S}$ and $\varphi_i\sim\mathcal{N}_p$.

We conclude that with this choice of $\gamma$, with probability at least $1-6\log(u)e^{-\frac{\gamma^2p^3n}{6^2\cdot 3}}$
\begin{align*}
    (1-\gamma)\mathop{\operatorname{E}}_{\substack{h\sim\mathcal{F},\\s\sim\mathcal{S},\\\varphi_i\sim\mathcal{N}_p}}[Z_i]\le  Z_i\le (1+\gamma)\mathop{\operatorname{E}}_{\substack{h\sim\mathcal{F},\\s\sim\mathcal{S},\\\varphi_i\sim\mathcal{N}_p}}[Z_i].
\end{align*}
\end{proof}

\subsection{Size of interval for input size}
Before proving Lemma \ref{lem:accuracy}, we give a technical lemma:
\begin{lemma}
\label{lem:technicalintervalprod}
For any $0<\gamma<\frac{1}{\frac{2e^3}{1-2p}-1}$ any value
\begin{align}
\label{int:targetlemma}
\hat{m}\in\left[2^in\ln\left(\frac{1-2p}{1-\frac{2Z_i}{(1+\gamma)n}}\right),2^in\ln\left(\frac{1-2p}{1-\frac{2Z_i}{(1-\gamma)n}}\right)\right]
\end{align}
satisfies 
\begin{align*}
    \hat{m}&\ge \left(1-\eta\right)2^in\ln\left(\frac{1}{\prod_{j\in A}\left(1-\frac{w_j}{2^in}\right)}\right)\\
    \hat{m}&\le \left(1+\eta\right)2^in\ln\left(\frac{1}{\prod_{j\in A}\left(1-\frac{w_j}{2^in}\right)}\right)
\end{align*}
for 
\[
\eta = \frac{6\gamma\left(\frac{e^{3}}{1-2p}-1\right)}{(1-\gamma)-2\gamma\left(\frac{e^{3}}{1-2p}-1\right)}
\]
with probability at least $1-6\log(u)e^{-\gamma^2p^3n/108}$ for the $i$ where $\frac{\Vert w\Vert_1}{2^in}\in[1,2]$.
\end{lemma}
\begin{proof}
By Lemma \ref{lem:expectations}
\[
\prod_{j\in A}\left(1-\frac{w_j}{2^in}\right)=\frac{1-\frac{2\operatorname{E}[Z_i]}{n}}{1-2p}
\]
and so by Lemma \ref{lem:concentrated}, with probability at least $1-6\log(u)e^{-\gamma^2p^3n/108}$ we have for any $0<\gamma<1$ that for  all $i=0,...,\log(u)-1$ simultaneously.
\begin{align}
\label{interval:prodforZ_i}
\frac{1-\frac{2Z_i}{(1-\gamma)n}}{1-2p}<\prod_{j\in A}\left(1-\frac{w_j}{2^in}\right)<\frac{1-\frac{2Z_i}{(1+\gamma)n}}{1-2p}.
\end{align}
For convenience, we consider the slightly bigger interval -- note that if (\ref{interval:prodforZ_i}) is satisfied, then so is this interval: 
\begin{align*}
\frac{1-\frac{2(1+\gamma)\operatorname{E}[Z_i]}{(1-\gamma)n}}{1-2p}<\prod_{j\in A}\left(1-\frac{w_j}{2^in}\right)<\frac{1-\frac{2(1-\gamma)\operatorname{E}[Z_i]}{(1+\gamma)n}}{1-2p},
\end{align*}
where the left-hand side can be reordered as
\begin{align}
\label{interval:prodforexpectationZ_ilower}
\left(1-\frac{2\gamma}{1-\gamma}\left(\frac{1}{(1-2p)\prod_{j\in A}\left(1-\frac{w_j}{2^in}\right)}-1\right)\right)\prod_{j\in A}\left(1-\frac{w_j}{2^in}\right)
\end{align}
and the right-hand side as
\begin{align}
\label{interval:prodforexpectationZ_iupper}
\left(1+\frac{2\gamma}{1+\gamma}\left(\frac{1}{(1-2p)\prod_{j\in A}\left(1-\frac{w_j}{2^in}\right)}-1\right)\right)\prod_{j\in A}\left(1-\frac{w_j}{2^in}\right).
\end{align}
We will bound this interval further using the following claim:
\begin{claim}
Define 
\[
\beta^*:=\frac{2\gamma}{1-\gamma}\left(\frac{e^{2+\frac{1}{2^{i-1}n}}}{1-2p}-1\right).
\]
Whenever $\frac{\Vert w\Vert_1}{2^in}< 2$, the interval defined by  (\ref{interval:prodforexpectationZ_ilower}) and (\ref{interval:prodforexpectationZ_iupper}) is contained in 
\begin{align*}
    \left[\left(1-\beta^*\right)\prod_{j\in A}\left(1-\frac{w_j}{2^in}\right),\left(1+\beta^*\right)\prod_{j\in A}\left(1-\frac{w_j}{2^in}\right)\right]
\end{align*}
\end{claim}
\begin{proof}[Proof of Claim]
As $\frac{2\gamma}{1+\gamma}<\frac{2\gamma}{1-\gamma}$, we increase (\ref{interval:prodforexpectationZ_iupper}) to 
\[
\left(1+\frac{2\gamma}{1-\gamma}\left(\frac{1}{(1-2p)\prod_{j\in A}\left(1-\frac{w_j}{2^in}\right)}-1\right)\right)\prod_{j\in A}\left(1-\frac{w_j}{2^in}\right).
\]
Observing that when $\frac{\Vert w\Vert_1}{2^in}\le 2$
\begin{align*}
    \frac{2\gamma}{1-\gamma}\left(\frac{1}{(1-2p)\prod_{j\in A}\left(1-\frac{w_j}{2^in}\right)}-1\right)&\le \frac{2\gamma}{1-\gamma}\left(\frac{e^{\frac{\Vert w\Vert_1}{2^in}+\frac{\Vert w\Vert_1}{(2^in)^2}}}{1-2p}-1\right)\\
    &\le \frac{2\gamma}{1-\gamma}\left(\frac{e^{2+\frac{1}{2^{i-1}n}}}{1-2p}-1\right)=:\beta^*
\end{align*}
we have the result.
\end{proof}

Applying the claim, we consider the interval:
\begin{align}
\label{interval:prodbetastarlower}
    &2^in\ln\left(\frac{1}{\prod_{j\in A}\left(1-\frac{w_j}{2^in}\right)}\right)\ge2^in\ln\left(\frac{1}{(1+\beta^*)\prod_{j\in A}\left(1-\frac{w_j}{2^in}\right)}\right)\\
    \label{interval:prodbetastarupper}
    &2^in\ln\left(\frac{1}{\prod_{j\in A}\left(1-\frac{w_j}{2^in}\right)}\right)\le 2^in\ln\left(\frac{1}{(1-\beta^*)\prod_{j\in A}\left(1-\frac{w_j}{2^in}\right)}\right).
\end{align}
We remind the reader that by construction, this interval contains the target interval (\ref{int:targetlemma}).

We consider the ratio between the end-points of the interval defined by (\ref{interval:prodbetastarlower}) and (\ref{interval:prodbetastarupper}).
Observe that 
\begin{align*}
    \frac{2^in\ln\left(\frac{1}{(1-\beta^*)\prod_{j\in A}\left(1-\frac{w_j}{2^in}\right)}\right)}{2^in\ln\left(\frac{1}{(1+\beta^*)\prod_{j\in A}\left(1-\frac{w_j}{2^in}\right)}\right)}&=\frac{\ln\left(\frac{1}{\prod_{j\in A}\left(1-\frac{w_j}{2^in}\right)}\right)-\ln(1-\beta^*)}{\ln\left(\frac{1}{\prod_{j\in A}\left(1-\frac{w_j}{2^in}\right)}\right)-\ln(1+\beta^*)}\\
    &\le \frac{\ln\left(\frac{1}{\prod_{j\in A}\left(1-\frac{w_j}{2^in}\right)}\right)+\frac{\beta^*}{1-\beta^*}}{\ln\left(\frac{1}{\prod_{j\in A}\left(1-\frac{w_j}{2^in}\right)}\right)-\beta^*}\\&=1+\frac{\beta^*\left(1+\frac{1}{1-\beta^*}\right)}{\ln\left(\frac{1}{\prod_{j\in A}\left(1-\frac{w_j}{2^in}\right)}\right)-\beta^*}
\end{align*}
where the inequality follows from 
\[
\frac{x}{1+x}\le \ln(1+x)\le x,\qquad x>-1.
\]
For $\beta^*<1/2$, we have
\begin{align*}
    \frac{\beta^*\left(1+\frac{1}{1-\beta^*}\right)}{\ln\left(\frac{1}{\prod_{j\in A}\left(1-\frac{w_j}{2^in}\right)}\right)-\beta^*}&<\frac{3\beta^*}{\ln\left(\frac{1}{\prod_{j\in A}\left(1-\frac{w_j}{2^in}\right)}\right)-\beta^*}\\
    &<\frac{3\beta^*}{\frac{\Vert w\Vert_1}{2^in}-\beta^*}
\end{align*}
Observe that as $\frac{\Vert w\Vert_1}{2^in}$ increases, it gets easier to satisfy this inequality. But we remind ourselves of the Claim, where we required $\frac{\Vert w\Vert_1}{2^in}< 2$. So the interval in (\ref{interval:prodbetastarlower}) and (\ref{interval:prodbetastarupper}) does not necessarily contain the target interval (\ref{int:targetlemma}) for larger values of $\frac{\Vert w\Vert_1}{2^in}$. Assume further that $\frac{\Vert w\Vert_1}{2^in}\ge 1$. Then
\begin{align*}
    \frac{3\beta^*}{\frac{\Vert w\Vert_1}{2^in}-\beta^*}<\frac{3\beta^*}{1-\beta^*}.
\end{align*}
So, we conclude that with probability at least $1-6\log(u)e^{-\gamma^2p^3n/108}$, any value in the target interval (\ref{int:targetlemma}) is within a factor $1+\frac{3\beta^*}{1-\beta^*}$ of $2^in\ln\left(\left(\prod_{j\in A}\left(1-\frac{w_j}{2^in}\right)\right)^{-1}\right)$. 

Inserting the value of $\beta^*$, we obtain an estimate within a factor of
\[
1+\frac{6\gamma\left(\frac{e^{2+\frac{1}{2^{i-1}n}}}{1-2p}-1\right)}{(1-\gamma)-2\gamma\left(\frac{e^{2+\frac{1}{2^{i-1}n}}}{1-2p}-1\right)}<1+\frac{6\gamma\left(\frac{e^{3}}{1-2p}-1\right)}{(1-\gamma)-2\gamma\left(\frac{e^{3}}{1-2p}-1\right)}.
\]
Thus it suffices that
\[
\gamma<\frac{1}{\frac{2e^3}{1-2p}-1}.
\]
\end{proof}

We are now ready to prove Lemma \ref{lem:accuracy}:
\label{app:accuracy}
\accuracy*
\begin{proof}
We will choose $\gamma$ in terms of the accuracy parameter $\beta$, such that with high probability any estimate from the interval
\begin{align}
\label{interval:target}
\left[2^in\ln\left(\frac{1-2p}{1-\frac{2Z_i}{(1+\gamma)n}}\right),2^in\ln\left(\frac{1-2p}{1-\frac{2Z_i}{(1-\gamma)n}}\right)\right]
\end{align}
is within a factor $(1+\beta)$ of $\Vert w\Vert_1$.
We do this in a few steps:
First, we show that any value from (\ref{interval:target}) is a good estimate of
\begin{align}
\label{expr:productexprforestim}
    2^in\ln\left(\frac{1}{\prod_{j\in A}\left(1-\frac{w_j}{2^in}\right)}\right).
\end{align}
As $2^in\ln\left(\frac{1}{e^{-\frac{\Vert w\Vert_1}{2^in}}}\right)=\Vert w\Vert_1$ and
\begin{align*}
    \frac{2^in\ln\left(\frac{1}{\prod_{j\in A}\left(1-\frac{w_j}{2^in}\right)}\right)}{2^in\ln\left(\frac{1}{e^{-\frac{\Vert w\Vert_1}{2^in}}}\right)}\le \frac{\ln\left(e^{\frac{\Vert w\Vert_1}{2^in}+\frac{\Vert w\Vert_1}{(2^in)^2}}\right)}{\ln\left(e^{\frac{\Vert w\Vert_1}{2^in}}\right)}=1+\frac{1}{2^in}
\end{align*}
where we used the Taylor expansion of the exponential function, we have
\[
\Vert w\Vert_1\le 2^in\ln\left(\frac{1}{\prod_{j\in A}\left(1-\frac{w_j}{2^in}\right)}\right)\le \left(1+\frac{1}{2^in}\right)\Vert w\Vert_1.
\]
So a good estimate for (\ref{expr:productexprforestim}) will allow for a good estimate of $\Vert w\Vert_1$.
The technical lemma, Lemma \ref{lem:technicalintervalprod}, shows that as long as $\Vert w\Vert_1$ is sufficiently large, that is, there is an $i$ such that $\frac{\Vert w\Vert_1}{2^in}\in[1,2)$, we get a suitable estimate for (\ref{expr:productexprforestim}) with the interval (\ref{interval:target}) with high probability.

Hence, any value from (\ref{interval:target}) is within a factor $(1+\beta)$ of $\Vert w\Vert_1$ for 
\[
\gamma < \frac{(\beta-1/n)(1-2p)}{7e^3}<\frac{\beta-1/n}{7\left(\frac{e^3}{1-2p}-1\right)}< \frac{\beta-\frac{1}{2^in}}{7\left(\frac{e^3}{1-2p}-1\right)}
\]
for $\beta>\frac{1}{n}$. We will choose $n$ in terms of $\beta$ such that this is always satisfied.
Clearly, this value of $\gamma$ is significantly smaller than the requirement from Lemma \ref{lem:technicalintervalprod}, which concludes the proof.
\end{proof}
\end{document}